\documentclass[11pt, a4paper]{article}

\usepackage{setspace}
\usepackage{bm}
\usepackage{upgreek}
\usepackage{rotating}
\usepackage{graphicx}
\usepackage{verbatim}
\usepackage{wrapfig}
\usepackage[margin=2.5cm]{geometry}
\usepackage{amsmath}
\usepackage{placeins}
\usepackage{url}

\usepackage{cite}

\usepackage{fancyhdr}
\usepackage{color}
\usepackage{amsthm}
\usepackage{amsmath, amssymb}
\usepackage{mathabx}
\usepackage{longtable}
\usepackage{lscape}
\usepackage{multicol}
\usepackage{pdfpages}
\usepackage{siunitx}
\usepackage{graphicx}
\usepackage{enumerate}

\usepackage[affil-it]{authblk}

\usepackage{caption}
\usepackage{subcaption}

\usepackage{bbm}
\usepackage{nameref}

\usepackage{tikz}
\usetikzlibrary{positioning}
\tikzset{>=stealth}
\usetikzlibrary{decorations.markings}
\usetikzlibrary{patterns}
\usetikzlibrary{math}

\usepackage{hyperref}

\usepackage{chngcntr}
\counterwithout{table}{section}

\newtheorem{all}{Theorem}[section]

\theoremstyle{plain}
\newtheorem{lemma}[all]{Lemma}
\newtheorem{thm}[all]{Theorem}

\newtheorem{prop}[all]{Proposition}

\theoremstyle{definition}

\newtheorem{rem}[all]{Remark}

\newcommand{\BI}{{\mathbb{I}}}

\newcommand{\BR}{{\mathbb{R}}}

\newcommand{\BZ}{{\mathbb{Z}}}

\newcommand{\dd}{{\mathrm{d}}}
\newcommand{\p}{\partial}

\newcommand{\com}[1]{}

\numberwithin{equation}{section}
\begin{document}
\title{Boundary Superconductivity in the BCS Model}
\author[1]{Christian Hainzl\thanks{hainzl@math.lmu.de}}
\author[2]{Barbara Roos\thanks{barbara.roos@ist.ac.at}}
\author[2]{Robert Seiringer\thanks{robert.seiringer@ist.ac.at}}
\affil[1]{Mathematisches Institut, Ludwig-Maximilians-Universit\"at M\"unchen, 80333 Munich, Germany}
\affil[2]{IST Austria, Am Campus 1, 3400 Klosterneuburg, Austria}

\date{\today}

\maketitle

\begin{abstract}
We consider the linear BCS equation, determining the BCS critical temperature, in the presence of a boundary, where Dirichlet boundary conditions are imposed. In the one-dimensional case with point interactions, we prove that the critical temperature is strictly larger than the bulk value, at least at weak coupling. In particular, the Cooper-pair wave function localizes near the boundary, an effect that cannot be modeled by effective Neumann boundary conditions on the order parameter as often imposed in Ginzburg--Landau theory.
We also show that the relative shift in critical temperature vanishes if the coupling constant either goes to zero or to infinity.
\end{abstract}

%

\section{Introduction and Main Result}\label{sec:intro}
We study how a boundary influences the critical temperature of a superconductor in the Bardeen--Cooper--Schrieffer (BCS) model.
At superconductor--insulator (or superconductor--vacuum) boundaries, it is natural to impose Dirichlet boundary conditions on the Cooper-pair wave function.
In several works \cite{caroli_sur_1962,de_gennes_boundary_1964,abrikosov_concerning_1964} it was concluded that the presence of the boundary only affects the Cooper-pair wave function on microscopic scales; in particular, on larger scales described by Ginzburg--Landau theory (GL), the effect of the Dirichlet boundary conditions disappears and consequently the GL order parameter  should satisfy {\em Neumann} boundary conditions \cite[Ch.~7.3]{gennes_superconductivity_1999}, \cite[Ch.~6]{parks_superconductivity_1969}. 
This seems to implicitly assume that the effect of the boundary on the critical temperature is negligible.
Recent computations \cite{samoilenka_boundary_2020,benfenati_boundary_2021,barkman_elevated_2022} indicate, however, that the Cooper-pair wave function can localize near the boundary, leading to an increase in the critical temperature compared to its bulk value.
In this paper, we shall give a rigorous proof of the occurrence of this phenomenon in the simplest setting of one dimension, with $\delta$-interactions among the particles. 
We consider a system on the half-line, where the boundary is then just a point.

The increase of the critical temperature in the presence of a boundary has some far-reaching implications.
First of all, it implies that  boundary superconductivity in the BCS model sets in already above the bulk value of the critical temperature. 
Second, it questions the validity of the often employed phenomenological GL theory in the presence of boundaries, as detailed in  \cite{samoilenka_microscopic_2021}.
Note that GL theory has so far only been rigorously derived from the BCS model for periodic systems without boundaries \cite{frank_microscopic_2011}. 
(In the low-density BEC limit at zero temperature it was shown in \cite{frank_condensation_2017} that the effective Gross--Pitaevskii theory inherits the microscopic Dirichlet boundary conditions.) 

In mathematical terms,  the presence of a boundary manifests itself in a compact perturbation of a translation-invariant operator, and we shall show that at weak coupling this leads to the appearance of discrete eigenvalues outside the continuous spectrum. In particular, there is an effective attraction to the boundary, which is strong enough  to create bound states. 

In the following, we shall consider a superconductor on a domain $\Omega$, with either $\Omega=\BR$ or $\Omega=\BR_+ =(0,\infty)$.
The main quantity of interest is the linear two-particle operator
\begin{equation}\label{H_original}
H^\Omega_T= \frac{-\Delta_x-\Delta_y-2\mu}{\tanh\left(\frac{-\Delta_x-\mu}{2T}\right)+\tanh\left(\frac{-\Delta_y-\mu}{2T}\right)}-v \delta(x-y)
\end{equation}
acting in $L_{\rm{symm}}^2(\Omega^2)=\{\psi \in L^2(\Omega^2) \vert \psi(x,y)=\psi(y,x)\ \mathrm{for\ all}\ x,y\in \Omega\}$, where $\Delta$ denotes the {\em Dirichlet} Laplacian on $\Omega$, and the subscripts $x$ and $y$, respectively, indicate the variable on which $\Delta$ acts. 
The first term is defined through functional calculus.
In the second term, $\delta$ is the Dirac delta distribution, and $v>0$ is a coupling constant. 
Moreover, $T>0$ denotes the temperature, and $\mu\in \BR$ is the chemical potential.

As explained in \cite{frank_bcs_2019}, $H^\Omega_T$ characterizes the local stability of the normal state in BCS theory.
If $ H^\Omega_T$ has spectrum below zero, i.e.~$\inf \sigma (H^\Omega_T) <0$, the normal state is unstable and the system in $\Omega$ is superconducting.
If $\inf \sigma (H^\Omega_T) \geq 0$, the normal state is locally stable.
We define the critical temperatures $T_c^\Omega$ as 
\begin{equation}\label{def:Tc}
T_c^\Omega(v):=\inf \left\{T\in (0,\infty) \vert \inf \sigma(H_T^\Omega) \geq 0\right\} \,.
\end{equation}
The sample is thus superconducting for $T< T_c^\Omega$. In the translation-invariant case, i.e. $\Omega=\BR$, 
it is also known that local stability of the normal state implies global stability \cite{hainzl_bcs_2008}; in particular, 
the sample is always in a normal state for $T\geq T_c^\BR$ in this case, i.e.~$T_c^\BR$ separates the superconducting and the normal phases.
For the point interactions considered in \eqref{H_original}, one can derive the explicit relation
\begin{equation}\label{T_c^0_explicit}
\frac{1}{2\pi}\int_\BR \frac{\tanh\left(\frac{q^2-\mu}{2T^\BR_c(v)}\right)}{q^2-\mu} \,\dd q=\frac{1}{v} \,.
\end{equation}

Because of translation invariance, $H_T^\BR$ has purely essential spectrum.
Moreover, $H_T^{\BR_+}$ has the same essential spectrum and possibly  additional eigenvalues below it.
In particular, for all $v>0$ the critical temperatures satisfy
\begin{equation}\label{T1>T0}
T_c^{\BR_+}(v)\geq T_c^\BR(v).
\end{equation}
Our main result states that this inequality is actually strict, at least for small $v$, proving that the boundary increases the critical temperature.
Moreover, the relative difference between the two critical temperatures vanishes both in the weak and in the strong coupling limit.

\begin{thm}\label{main_result}
Let $\mu>0$.
\begin{enumerate}[(i)]
\item \label{thm_ex}There is a $\tilde v>0$ such that 
\begin{equation}
T_c^{\BR_+}(v)>T_c^\BR(v)
\end{equation}
for $0<v<\tilde v$.
\item \label{thm_weak}
In the weak coupling limit
\begin{equation}
\lim_{v\to 0} \frac{T_c^{\BR_+}(v)-T_c^\BR(v)}{T_c^\BR(v)}=0 
\end{equation}
\item \label{thm_strong} In the strong coupling limit
\begin{equation}
 \lim_{v\to \infty} \frac{T_c^{\BR_+}(v)-T_c^\BR(v)}{T_c^\BR(v)}=0 
\end{equation}
\end{enumerate}
\end{thm}

This result can be viewed as a rigorous justification of the observations in \cite{samoilenka_boundary_2020}. Numerics shows that the ratio  $T_c^{\BR_+}(v)/T_c^\BR(v)$ can be as large as $1.06$, see \cite[Fig.~2]{samoilenka_boundary_2020}. Moreover, numerics also suggests that  $T_c^{\BR_+}(v)$ and $T_c^\BR(v)$ actually agree for $v$ large enough, but it remains an open problem to show this.

Part \eqref{thm_ex} of Theorem~\ref{main_result} follows from the existence of an eigenvalue of $H^{\BR_+}_T$ below the spectrum of $H^\BR_T$.
It is quite remarkable that a Dirichlet boundary can decrease the ground state energy and create bound states. In contrast, for two-particle Schr\"odinger operators of the form $-\Delta_x -\Delta_y+V(x-y)$, only Neumann boundaries can bind states \cite{egger_bound_2020,roos_two-particle_2021}.

While we restrict our attention in this article to the one-dimensional setting with point interactions, we expect that our methods can be generalized to a larger class of interaction potentials, as well as to higher dimensions and the corresponding more complicated geometries possible. We shall leave these generalizations for future investigations, however.

\begin{rem}\label{main_result_neumann}
Our techniques can also be applied in case of Neumann boundary conditions for $\Delta$ on $\BR_+$. In this case one obtains the 
following results instead.
\begin{enumerate}[(i)]
\item \label{thm_n_ex} For all $v>0$
\begin{equation}
T_c^{\BR_+}(v)>T_c^\BR(v)
\end{equation}
\item 
In the weak coupling limit
\begin{equation}\label{thm_n_weak}
\lim_{v\to 0} \frac{T_c^{\BR_+}(v)-T_c^\BR(v)}{T_c^\BR(v)}=0 
\end{equation}
\item \label{thm_n_strong} In the strong coupling limit
\begin{equation}
0< \lim_{v\to \infty} \frac{T_c^{\BR_+}(v)-T_c^\BR(v)}{T_c^\BR(v)}<\infty
\end{equation}
\end{enumerate}
\end{rem}

In the remainder of this article we shall give the proof of Theorem~\ref{main_result}. In the next Section~\ref{sec:pre}, we shall use the Birman--Schwinger principle to conveniently reformulate the problem in terms of bounded operators and compact perturbations. Section~\ref{sec:ex} contains the proof of part (i), the existence of boundary superconductivity. The analysis of the weak and strong coupling limits in parts (ii) and (iii) is the content of Sections~\ref{sec:weak} and~\ref{sec:strong}, respectively. Finally, Section~\ref{sec:pf} contains the proofs of some auxiliary Lemmas.

\section{Preliminaries}\label{sec:pre}
Let us fix the notation
\begin{equation}
L_{T,\mu}(p,q):=\frac{\tanh\left(\frac{p^2-\mu}{2T}\right)+\tanh\left(\frac{q^2-\mu}{2T}\right)}{p^2 + q^2 - 2\mu}.
\end{equation}
Using the partial fraction expansion for $\tanh$ (Mittag-Leffler series), one can obtain the series representation \cite{frank_bcs_2019}
\begin{equation}\label{L_series}
L_{T,\mu}(p,q)={2T}\sum_{n\in \BZ}\frac{1}{p^2 -\mu-iw_n} \frac{1}{q^2 -\mu+iw_n}
\end{equation}
for $w_n=\pi (2n+1)T$. 
Moreover, let
\begin{equation}
F_{T,\mu}(p):=L_{T,\mu}(p,p)=\frac{\tanh\left(\frac{p^2-\mu}{2T}\right)}{p^2 -\mu}
\end{equation}
and
\begin{equation}\label{def:B}
B_{T,\mu}(p,q):=L_{T,\mu}\left(\frac{p+q}{2},\frac{p-q}{2}\right)
\end{equation}

In order to control the kinetic energy in $H^\Omega_T$ the following bounds turn out to be useful. We shall prove them in Section~\ref{sec:pfpre}.

\begin{lemma}\label{lea:L_laplace}
Let $T>0$.
There are constants $C_1(T,\mu),C_2(T,\mu)>0$ such that for all $p,q\in \BR$ 
\begin{equation}
C_1(T,\mu)(1+p^2+q^2)\leq L_{T,\mu}(p,q)^{-1} \leq C_2(T,\mu)(1+p^2+q^2)
\end{equation}
Moreover, for $T_0>0$ there is a $C_3(T_0,\mu)>0$ such that
\begin{equation}
C_3(T_0,\mu) (T+p^2+q^2)\leq L_{T,\mu}(p,q)^{-1}
\end{equation}
for all $T>T_0$ and $p,q\in \BR$.
\end{lemma}

Since $v \delta(x-y)$ is infinitesimally form bounded with respect to $-\Delta_x-\Delta_y$, it follows that the $H_T^\Omega$ are self-adjoint operators defined via the KLMN theorem.
Moreover, the operators $H_T^\Omega$ become positive for $T$ large enough.
In particular, the critical temperatures defined in \eqref{def:Tc} are finite in both cases $\Omega=\BR$ and $\Omega=\BR_+$. 

Let $L_{T,\mu}^\Omega$ denote the operator $L_{T,\mu}(-i\nabla_x,-i \nabla_y)$ defined through functional calculus.
Of course, $L^\Omega_{T,\mu}$ depends on the domain $\Omega$ and on the boundary conditions imposed on $\Delta$.
Its integral kernel is given by
\begin{equation}
L^\Omega_{T,\mu}(x,y;x',y')= \int_{\BR^2} \dd p\, \dd q\, \overline{t_\Omega(xp)}\overline{t_\Omega(yq)}L_{T,\mu}(p,q)t_\Omega(x'p)t_\Omega(y'q) \,,
\end{equation}
where for the problem on the full real line $t_\BR(x)= \frac{1}{\sqrt{2\pi}}e^{-ix}$ and on the half-line with Dirichlet boundary condition $t_{\BR_+}(x)= \frac{1}{\sqrt{\pi}}\sin(x)$.
For Neumann boundary conditions, one would have $t_{\BR_+}(x)= \frac{1}{\sqrt{\pi}}\cos(x)$ instead.

It is convenient to switch to the Birman--Schwinger formulation of the problem.
For a more regular interaction $V$ instead of $\delta$, the Birman-Schwinger operator would be $V^{1/2}L^\Omega_{T,\mu}V^{1/2}$.
For the $\delta$-case, it turns out that $V^{1/2}$ has to be understood as restriction of a two-body wave function to its diagonal.
Hence, the Birman-Schwinger operator has kernel $L^\Omega_{T,\mu}(x,x;x',x')$ and acts on functions of one variable only.
For the two domains under consideration, the Birman-Schwinger operators $A_{T,\mu}^{\BR_+}:L^2((0,\infty))\to L^2((0,\infty))$ and $A_{T,\mu}^\BR:L^2(\BR)\to L^2(\BR)$ are explicitly given by
\begin{equation}\label{A^1_xy}
(A_{T,\mu}^{\BR_+} \alpha)(x)=\frac{1}{\pi^2}\int_{\BR}\dd p \int_{\BR} \dd q  \int_0^\infty \dd y\sin(px)\sin(qx) L_{T,\mu}(p,q)\sin(py)\sin(qy)\alpha(y)
\end{equation}
and
\begin{equation}
(A_{T,\mu}^\BR \beta)(x)=\frac{1}{4\pi^2}\int_\BR \dd p \int_\BR \dd q  \int_\BR \dd y\, e^{i(p+q)(x-y)} L_{T,\mu}(p,q)\beta(y)
\end{equation}

\begin{lemma}\label{bs_condition} 
The condition $\inf \sigma( H_T^\Omega) <0$ is equivalent to 
\begin{equation}
\sup \sigma(A_{T,\mu}^\Omega)>\frac{1}{v}
\end{equation}
for either $\Omega=\BR$ or $\Omega=\BR_+$. 
\end{lemma}
\begin{proof}
The quadratic form corresponding to $H^\Omega_T$ is defined on the Sobolev space $D_\Omega=H^1_0(\Omega^2)$.
Since the operator $L^\Omega_{T,\mu}$ is positive definite, one can write
\begin{equation}
H^\Omega_T=\left( L^\Omega_{T,\mu} \right)^{-1} - v \delta = \frac{1}{\sqrt{L^\Omega_{T,\mu}}}\left(\BI - v \sqrt{L^\Omega_{T,\mu}} \delta \sqrt{L^\Omega_{T,\mu}}\right)\frac{1}{\sqrt{L^\Omega_{T,\mu}}} \,.
\end{equation}
Hence, $\inf \sigma( H_T^\Omega) <0$ is equivalent to 
\begin{equation}\label{ldlv}
\sup_{\psi \in (L^\Omega_{T,\mu})^{-1/2}D_\Omega, \lVert \psi \rVert_2=1}\left\langle \psi \left|  \sqrt{L^\Omega_{T,\mu}} \delta \sqrt{L^\Omega_{T,\mu}} \right| \psi \right\rangle>\frac{1}{v} \,.
\end{equation}
By Lemma~\ref{lea:L_laplace}, $\sqrt{L^\Omega_{T,\mu}}:L^2(\Omega^2) \to D_\Omega$ and its inverse are bounded.
Hence, $(L^\Omega_{T,\mu})^{-1/2}D_\Omega=L^2(\Omega^2)$.
The projection onto the diagonal $H^1(\Omega^2)\to L^2(\Omega)$, $\psi(x,y) \mapsto \psi(x,x)$ defines a bounded operator \cite[Thm 4.12]{adams_sobolev_2003}.
Let $M_\Omega:L^2(\Omega^2)\to L^2(\Omega)$ be the composition of $\sqrt{L^\Omega_{T,\mu}}$ with the projection $H^1(\Omega^2)\to L^2(\Omega)$.
Explicitly, $M_\Omega$ is given by
\begin{equation}
M_\Omega \psi (x)= \int_{\BR^2} \dd p \, \dd q \int_{\Omega^2} \dd x' \dd y' \, \overline{t_\Omega(xp)}\overline{t_\Omega(xq)}\sqrt{L_{T,\mu}(p,q)}t_\Omega(x'p)t_\Omega(y'q)\psi(x',y')
\end{equation}
where $t_\BR(x)= \frac{1}{\sqrt{2\pi}}e^{-ix}$ and $t_{\BR_+}(x)= \frac{1}{\sqrt{\pi}}\sin(x)$.
Note that $\sqrt{L^\Omega_{T,\mu}} \delta \sqrt{L^\Omega_{T,\mu}} = M_\Omega^\dagger M_\Omega$ and  $A_{T,\mu}^\Omega=M_\Omega M_\Omega^\dagger$.
Hence, $\sigma(A_{T,\mu}^\Omega)\setminus\{0\}= \sigma(\sqrt{L^\Omega_{T,\mu}} \delta \sqrt{L^\Omega_{T,\mu}})\setminus\{0\}$ and the claim follows.
\end{proof}

From now on we will work with the operators $A_{T,\mu}^\Omega$ rather than $H_T^\Omega$.
In momentum space, the operator $A_{T,\mu}^\BR$ is multiplication by the function
\begin{equation}\label{intBdq}
A_{T,\mu}(p)=\frac{1}{4\pi} \int_\BR B_{T,\mu}(p,q) \dd q \,,
\end{equation}
where $B$ is defined in \eqref{def:B}.

\begin{lemma}[Momentum representation of $A_{T,\mu}^\BR$]\label{A_T_0_momentum}
With $\widehat \beta(p)=\frac{1}{\sqrt{2\pi}}\int_\BR \beta(x) e^{ip x} \dd x$ we have for all $\beta_1,\beta_2\in D(A_{T,\mu}^\BR)$ 
\begin{equation}\label{bsop_a0}
\langle \beta_1 | A_{T,\mu}^\BR | \beta_2\rangle =\int_\BR  \overline{\widehat \beta_1}(p) A_{T,\mu}(p) \widehat \beta_2 (p)\dd p.
\end{equation}
\end{lemma}

 The following Lemma shows that adding the boundary to the system effectively introduces the perturbation $\frac{1}{4\pi} B_{T,\mu}$, where $B_{T,\mu}$ is short for the operator with integral kernel $B_{T,\mu}(p,q)$.
 
 \begin{lemma}[Momentum representation of $A_{T,\mu}^{\BR_+}$]\label{A_T_momentum}
With $\widehat \alpha(p)=\int_0^\infty \alpha(x) \frac{1}{\sqrt{\pi}}\cos(px) \dd x$ we have for all $\alpha_1,\alpha_2\in D(A_{T,\mu}^{\BR_+})$ 
\begin{equation}\label{bsop_ak}
\langle \alpha_1 |  A_{T,\mu}^{\BR_+} | \alpha_2 \rangle=\int_\BR   \overline{\widehat \alpha_1}(p) A_{T,\mu}(p) \widehat \alpha_2 (p)\dd p-\frac{1}{4\pi}\int_\BR \int_\BR  \overline{\widehat \alpha_1}(p)  B_{T,\mu}\left(p,q\right)\widehat \alpha_2(q) \dd p \dd q.
\end{equation}
\end{lemma}
Note that here we work with the cosine transform and not the sine transform as might be expected from \eqref{A^1_xy}.
This is because $\alpha$ is the diagonal of a function which is antisymmetric under both $x\to -x$ and $y \to -y$ and hence symmetric under $(x,y) \to (-x,-y)$.
\begin{proof}[Proof or Lemma~\ref{A_T_momentum}]
Using that $\sin(px)\sin(qx)=\frac{1}{2}[\cos((p-q)x)-\cos((p+q)x)]$ and substituting $p'=p-q$ and $q'=p+q$ gives
\begin{multline}
\langle\alpha_1| A_{T,\mu}^{\BR_+} |\alpha_2\rangle=\frac{1}{\pi^2} \int_{\BR^2} \dd p\, \dd q \int_0^\infty \dd x \int_0^\infty \dd y\,  \overline{\alpha_1(x)}\sin(px)\sin(qx) L_{T,\mu}(p,q)\sin(py)\sin(qy)\alpha_2(y)\\
=\frac{1}{4\pi^2} \int_{\BR^2} \frac{\dd p' \dd q'}{2}  \int_0^\infty \dd x \int_0^\infty \dd y\,  \overline{\alpha_1(x)}[\cos(p'x)-\cos(q'x)] L_{T,\mu}\left(\frac{p'+q'}{2},\frac{p'-q'}{2}\right)
[\cos(p'y)-\cos(q'y)]\alpha_2(y)\\
=\int_{\BR^2} \frac{\dd p' \dd q'}{8\pi}\left[\overline{\widehat \alpha_1(p')}-\overline{\widehat \alpha_1(q')}\right] B_{T,\mu}(p',q')\left[\widehat \alpha_2(p')-\widehat \alpha_2(q')\right].
\end{multline}
Since $B(p',q')=B(q',p')$, this reduces to 
\begin{equation}
\langle \alpha_1 | A_{T,\mu}^{\BR_+} | \alpha_2\rangle=\int_{\BR^2} \frac{\dd p' \dd q'}{4\pi}\,\overline{\widehat \alpha_1(p')} B_{T,\mu}(p',q')\left[\widehat \alpha_2(p')-\widehat \alpha_2(q')\right].
\end{equation}
\end{proof}
Lemma~\ref{A_T_0_momentum} follows from an analogous computation.

Since the operator $A_{T,\mu}^\BR$ is multiplication by the function \eqref{intBdq}, it has purely essential spectrum.
The perturbation $B_{T,\mu}$ in $A_{T,\mu}^{\BR_+}$ is Hilbert--Schmidt and thus compact.
Hence,  $\sigma( A_{T,\mu}^\BR )=\sigma_{\text{ess}}( A_{T,\mu}^{\BR_+} )$.
It follows that for all $T<T_c^\BR(v)$ we have $\sup\sigma( A_{T,\mu}^{\BR_+} )\geq \sup\sigma( A_{T,\mu}^\BR )>1/v$, which  implies \eqref{T1>T0}.

\begin{rem}
Choosing Neumann instead of Dirichlet boundary conditions amounts to changing the minus sign in \eqref{bsop_ak} into a plus sign. 
\end{rem}

It is possible to give a more explicit expression for $ \sup\sigma( A_{T,\mu}^\BR )$.
The following is proved in Section~\ref{sec:pfpre}.
\begin{lemma}\label{intBdq_max}
For all $p\in \BR$ 
\begin{equation}
\int_\BR B_{T,\mu}(p,q) \dd q \leq \int_\BR B_{T,\mu}(0,q) \dd q \,.
\end{equation}
\end{lemma}
Consequently, 
\begin{equation}
a_{T,\mu}:=\sup \sigma(A_{T,\mu}^\BR)=\frac{1}{4\pi} \int_\BR B_{T,\mu}(0,q) \dd q \,.
\end{equation}
Hence, in the translation invariant case superconductivity is equivalent to $a_{T,\mu}>\frac{1}{v}$ and the critical temperature is determined by \eqref{T_c^0_explicit}.
Note that $a_{T,\mu}$ is decreasing in $T$.
Therefore, $T_c^\BR(v)$ is a monotonically increasing function of $v$.

\section{Existence of Boundary Superconductivity}\label{sec:ex}
From now on we assume that $\mu>0$.
In this Section, we show that for weak coupling the half-line critical temperature is higher than the bulk critical temperature.
The idea is to prove that for $T$ below a threshold $T_0>0$ we have 
\begin{equation}\label{A1>a}
\sup \sigma(A_{T,\mu}^{\BR_+})>a_{T,\mu}\,.
\end{equation}
Then consider $v<\tilde v:= a_{T_0,\mu}^{-1}$.
We must have $T_c^\BR(v)<T_0$ by the monotonicity of $T_c^\BR(v)$.
By definition and continuity of $\inf \sigma(H_T^{\BR^+})$ in $T$, $\sup \sigma\left(A_{T_c^{\BR_+}(v),\mu}^{\BR_+}\right)=\frac{1}{v}=a_{T_c^\BR(v),\mu}$.
If $T_c^\BR(v)=T_c^{\BR_+}(v)$, we would get a contradiction to \eqref{A1>a}.
Thus, $T_c^\BR(v)\neq T_c^{\BR_+}(v)$ and, together with \eqref{T1>T0}, part \eqref{thm_ex} of Theorem~\ref{main_result} follows.

To prove \eqref{A1>a}, we use the variational principle with a trial function mimicking the ground state found in \cite{samoilenka_boundary_2020}.
We choose $\psi_\epsilon^\lambda(x)=e^{-\epsilon \vert x\vert} + \lambda g(x)$, where $\lambda\in\BR$ and the cosine Fourier transform 
$\widehat g(p)=\frac{1}{\sqrt{\pi}}\int_0^\infty g(x) \cos(px) \dd x$ is real, continuous and centered at $2\sqrt{\mu}$.

\begin{prop}\label{prop:deltainteraction}
Let $\widehat g(p)=e^{-(\vert p\vert-2\sqrt{\mu})^2/b}$ for some constant $b>0$.
For  $\mu>0$ there exists $T_0>0$ such that for $T<T_0$ 
$$\max_\lambda \lim_{\epsilon\to 0} \langle \psi_\epsilon^\lambda\vert  A_{T,\mu}^{\BR_+}-a_{T,\mu} \BI \vert \psi_\epsilon^\lambda\rangle >0.$$ 
\end{prop}

As discussed above, Theorem~\ref{main_result} \eqref{thm_ex} follows directly from Prop.~\ref{prop:deltainteraction}.

\begin{proof}
Let $h_\epsilon(x)=e^{-\epsilon \vert x \vert}$.
The cosine Fourier transform of the trial state is $\widehat \psi_\epsilon^\lambda (p)=\widehat h_\epsilon(p)+\lambda \widehat g(p)$, where $\widehat h_\epsilon(p)=\frac{1}{\sqrt{\pi}}\frac{\epsilon}{\epsilon^2+p^2}.$
We have $\lim_{\epsilon\to 0} \langle \psi_\epsilon^\lambda\vert  A_{T,\mu}^{\BR_+}-a_{T,\mu} \BI \vert \psi_\epsilon^\lambda\rangle=\lim_{\epsilon\to 0} \langle {h_\epsilon} \vert   A_{T,\mu}^{\BR_+}-a_{T,\mu} \BI \vert  {h_\epsilon}\rangle+2\lambda \lim_{\epsilon\to 0} \langle g \vert   A_{T,\mu}^{\BR_+}-a_{T,\mu} \BI \vert  {h_\epsilon}\rangle+\lambda^2 \langle  g \vert  A_{T,\mu}^{\BR_+}-a_{T,\mu} \BI \vert  g \rangle$.
In Lemma~\ref{asymptotic_gltag} we show $\langle  g \vert  A_{T,\mu}^{\BR_+}-a_{T,\mu} \BI \vert  g \rangle<0$.
Maximizing over $\lambda$ thus yields 
\begin{equation}\label{max_poly}
\max_\lambda \lim_{\epsilon\to 0} \langle \psi_\epsilon^\lambda\vert   A_{T,\mu}^{\BR_+}-a_{T,\mu} \BI \vert \psi_\epsilon^\lambda\rangle = \lim_{\epsilon\to 0} \langle {h_\epsilon} \vert   A_{T,\mu}^{\BR_+}-a_{T,\mu} \BI \vert {h_\epsilon}\rangle-\frac{\lim_{\epsilon\to 0} \langle g \vert   A_{T,\mu}^{\BR_+}-a_{T,\mu}\BI \vert {h_\epsilon}\rangle^2}{ \langle  g \vert   A_{T,\mu}^{\BR_+}-a_{T,\mu} \BI \vert  g \rangle}
\end{equation}

We now compute the two limits. 
Note that for bounded continuous functions $f$, we have $\lim_{\epsilon\to 0} \int_\BR \frac{1}{\sqrt{\pi}}\frac{\epsilon}{\epsilon^2+p^2} f(p) \dd p= \sqrt{\pi}f(0)$.
Moreover, for bounded functions $f$ such that $\lim_{p\to 0} \frac{f(p)}{p}$ exists, $\lim_{\epsilon\to 0} \int_\BR \frac{1}{\pi}\frac{\epsilon^2}{(\epsilon^2+p^2)^2} f(p) \dd p= \frac{1}{\pi}\lim_{p\to 0} \frac{f(p)}{p}$.
With the momentum space representation of $A_{T,\mu}^{\BR_+}$ in Lemma~\ref{A_T_momentum} we thus obtain 
\begin{multline}\label{l1part}
 \lim_{\epsilon\to 0} \langle  {h_\epsilon} \vert  A_{T,\mu}^{\BR_+}-a_{T,\mu} \BI \vert  g \rangle 
=\lim_{\epsilon\to 0}\int_\BR \dd p\ \widehat h_\epsilon(p) \widehat g (p)\left(A_{T,\mu}(p)- A_{T,\mu}(0)\right)\\
-\lim_{\epsilon\to 0}\int_\BR \dd p\, \widehat h_\epsilon(p)\int_\BR \dd q\  \frac{1}{4\pi}B_{T,\mu}(p,q)\widehat g(q)
=-\frac{1}{4\sqrt{\pi}}\int_\BR \dd q\  B_{T,\mu}(0,q)\widehat g(q).
\end{multline}
Moreover, 
\begin{multline}\label{ex_hh}
 \lim_{\epsilon\to 0} \langle  {h_\epsilon} \vert  A_{T,\mu}^{\BR_+}-a_{T,\mu} \BI \vert {h_\epsilon} \rangle \\
=\lim_{\epsilon\to 0}\int_\BR \dd p \,\widehat h_\epsilon ^2(p)\int_\BR \dd q\  \frac{1}{4\pi} (B_{T,\mu}(p,q)- B_{T,\mu}(0,q))
-\lim_{\epsilon\to 0}\int_\BR \dd p\, \widehat h_\epsilon(p) \int_\BR \dd q\  \frac{1}{4\pi} B_{T,\mu}(p,q) \widehat h_\epsilon (q)\\
=\frac{1}{\pi}\lim_{p\to 0}\frac{1}{p} \int_\BR \dd q\  \frac{1}{4\pi}  (B_{T,\mu}(p,q)- B_{T,\mu}(0,q))
-\frac{1}{4} B_{T,\mu}(0,0).
\end{multline}
In the first summand, we want to interchange limit and integration using dominated convergence.
The following Lemma is proved below.

\begin{lemma}\label{lea:ex_dom_conv}
The function $f(p,q)=\frac{1}{p}   (B_{T,\mu}(p,q)- B_{T,\mu}(0,q))$ 
\begin{enumerate}[(i)]
\item is continuous at $p=0$ and satisfies $f(0,q)=0$ for all $q$. \label{lea_ex_dom_conv_1}
\item There is a $g\in L^1(\BR)\cap L^\infty(\BR)$ such that $\vert f(p,q) \vert\leq g(q)$ for all $p$ and $q$. \label{lea_ex_dom_conv_2}
\end{enumerate}
\end{lemma}

By dominated convergence the first term on the right hand side of  \eqref{ex_hh} vanishes and thus $ \lim_{\epsilon\to 0} \langle h_\epsilon \vert  A_{T,\mu}^{\BR_+}-a_{T,\mu} \BI \vert h_\epsilon \rangle =-\frac{1}{4} B_{T,\mu}(0,0).$
Combining this with (\ref{max_poly}) and (\ref{l1part}) yields
\begin{equation}\label{varprincprop}
\max_\lambda \lim_{\epsilon\to 0} \langle \psi_\epsilon^\lambda\vert   A_{T,\mu}^{\BR_+}-a_{T,\mu} \BI \vert \psi_\epsilon^\lambda\rangle = -\frac{1}{4} B_{T,\mu}(0,0) -\frac{1}{16\pi}\frac{\left(\int_\BR B_{T,\mu}(0,q)\widehat g(q)\dd q\right)^2}{\langle g \vert A_{T,\mu}^{\BR_+}-a_{T,\mu}\vert g\rangle }
\end{equation}
For $T \to 0$ the term $ B_{T,\mu}(0,0)$ is bounded while the second summand diverges logarithmically, which is content of the following Lemma.
\begin{lemma}\label{asymptotic_gltag}
Let $\widehat g(p)=e^{-\frac{(\vert p\vert-2\sqrt{\mu})^2}{b}}$ for some $b>0$.
Then, 
\begin{enumerate}[(i)]
\item{$\frac{4}{\sqrt{\mu}}e^{-\frac{4\mu}{b}}<\lim_{T\to 0}\left(\ln \frac{\mu}{T}\right)^{-1}\int_\BR B_{T,\mu}(0,q)\widehat g(q) \dd q<\frac{4}{\sqrt{\mu}}$,} \label{fg}
\item{$0\geq \lim_{T\to 0}\left(\ln \frac{\mu}{T}\right)^{-1}\langle g \vert A_{T,\mu}^{\BR_+}-a_{T,\mu} \vert g\rangle>-\infty $.} \label{glag}
\end{enumerate}
\end{lemma}
Therefore, the last term in \eqref{varprincprop}  dominates for small $T$ and makes the right hand side positive.  This completes the proof of Prop.~\ref{prop:deltainteraction}. 
\end{proof}

\begin{rem}
For Neumann boundary conditions, one obtains $ \lim_{\epsilon\to 0} \langle h_\epsilon \vert  A_{T,\mu}^{\BR_+}-a_{T,\mu} \BI \vert h_\epsilon \rangle =\frac{1}{4} L_{T,\mu}(0,0)>0$.
Hence, the trial state $h_\epsilon$ suffices to prove $\sup \sigma(A_{T,\mu}^{\BR_+})>a_{T,\mu}$ for all $T>0$. 
\end{rem}

\begin{proof}[Proof of Lemma~\ref{lea:ex_dom_conv}]
Using \eqref{L_series} one obtains the series representation 
\begin{equation}
f(p,q)=\frac{T}{8}\sum_{n\in \BZ} \frac{8\mu p -p^3+2pq^2-16 i q w_n}{\left(\left(\frac{p+q}{2}\right)^2-\mu-iw_n\right)\left(\left(\frac{p-q}{2}\right)^2-\mu+iw_n\right)\left(\left(\frac{q}{2}\right)^2-\mu-iw_n\right)\left(\left(\frac{q}{2}\right)^2-\mu+iw_n\right)}
\end{equation}
where $w_n=(2n+1)\pi T$.
From this, claim~\eqref{lea_ex_dom_conv_1} is easy to see.
For part~\eqref{lea_ex_dom_conv_2}, note that by Lemma~\ref{lea:L_laplace}, $\vert f(p,q)\vert<\frac{C}{1+q^2}=:g_1(q)$ for $\vert p \vert>\sqrt{\mu}$. 
For $\vert p \vert<\sqrt{\mu}$,
\begin{multline}
\sup_{(p,q)\in \BR^2,\vert p \vert<\sqrt{\mu}} \frac{\vert 8\mu p -p^3+2pq^2 \vert }{\left \vert \left(\left(\frac{p+q}{2}\right)^2-\mu-iw_n\right)\left(\left(\frac{p-q}{2}\right)^2-\mu+iw_n\right)\right \vert}\\
\leq \sup_{(p,q)\in \BR^2,\vert p \vert<\sqrt{\mu}} \frac{ 8\mu \vert p\vert +\vert p\vert ^3+2\vert p\vert q^2 }{\sqrt{\left[\left(\frac{p+q}{2}\right)^2-\mu\right]^2+w_0^2}\sqrt{\left[\left(\frac{p-q}{2}\right)^2-\mu\right]^2+w_0^2}}=:c_1 <\infty
\end{multline}
and 
\begin{multline}
\sup_{(p,q)\in \BR^2} \frac{16 \vert q  w_n\vert }{\left \vert \left(\left(\frac{p+q}{2}\right)^2-\mu-iw_n\right)\left(\left(\frac{p-q}{2}\right)^2-\mu+iw_n\right)\right \vert}\\
\leq \sup_{(p,q)\in \BR^2} \frac{ 16 \vert q \vert }{\sqrt{\left[\left(\frac{\vert p \vert +\vert q\vert }{2}\right)^2-\mu\right]^2+w_0^2}}=:c_2 <\infty
\end{multline}
With these estimates, one obtains for $\vert p \vert<\sqrt{\mu}$
\begin{equation}
\vert f(p,q) \vert \leq  \frac{T(c_1+c_2)}{8} \sum_{n\in \BZ} \frac{1}{\left(\frac{q^2}{4}-\mu \right)^2+w_n^2}
\end{equation}
Using that the summands are decreasing in $n$, we can estimate the sum by an integral
\begin{multline}
\vert f(p,q) \vert \leq \frac{T(c_1+c_2)}{4} \left[ \frac{1}{\left(\frac{q^2}{4}-\mu \right)^2+w_0^2} +\int_{1/2}^\infty  \frac{1}{\left(\frac{q^2}{4}- \mu\right)^2+4\pi^2T^2 x^2} \dd x \right]\\
=\frac{T(c_1+c_2)}{4} \left[ \frac{1}{\left(\frac{q^2}{4}-\mu \right)^2+w_0^2} +\frac{\arctan \left(\frac{\vert \frac{q^2}{4}- \mu\vert}{\pi T}\right)}{2\pi T \vert \frac{q^2}{4}- \mu\vert}\right]=:g_2(q)
\end{multline}
Clearly, $g=\max\{g_1,g_2\}\in L^1(\BR)\cap L^\infty(\BR)$.
\end{proof}

The logarithmic divergence in Lemma~\ref{asymptotic_gltag} originates from the following asymptotics proved in Section~\ref{sec:pfex}.
\begin{lemma}\label{asymptotic_int_F}
Let $\mu>0$. As $T\rightarrow 0$ 
\begin{equation}\label{int_F_asy}
\int_\BR F_{T,\mu}(p) \dd p= \frac{2}{\sqrt{\mu}} \left(\ln \frac{\mu}{T}+\gamma+\ln \frac{8}{\pi}\right)+o(1)= \int_{-\sqrt{2\mu}}^{\sqrt{2\mu}} F_{T,\mu}(p) \dd p+O(1),
\end{equation}
where $\gamma$ denotes the Euler--Mascheroni constant.
\end{lemma}

\begin{proof}[Proof of Lemma~\ref{asymptotic_gltag}]
Part (\ref{fg}). 
On the interval $[-2\sqrt{2\mu},2\sqrt{2\mu}]$ the minimum of $\widehat g$ is $e^{-\frac{4\mu}{b}}$.
We estimate 
\begin{align}\nonumber 
\int_{-2\sqrt{2\mu}}^{2\sqrt{2\mu}} B_{T,\mu}(0,p) e^{-\frac{4\mu}{b}}\dd p & \leq \int_\BR  B_{T,\mu}(0,p) \widehat g(p) \dd p
\\ &\leq \int_{-2\sqrt{2\mu}}^{2\sqrt{2\mu}}   B_{T,\mu}(0,p) \dd p+\int_\BR \chi_{\vert p \vert>2\sqrt{2\mu}} \frac{e^{-\frac{(\vert p\vert -2\sqrt{\mu})^2}{b}}}{(p/2)^2-\mu}\dd p, \label{bound_fg}
\end{align}
where we used $\widehat g(k)\leq 1$ and $\tanh(x)\leq 1$.
The last summand is some constant independent of $T$.
Using that $B_{T,\mu}(0,p)= F_{T,\mu}(p/2)$ and Lemma~\ref{asymptotic_int_F} the asymptotic behavior for $T\to 0$ is
\begin{equation}
\int_{-2\sqrt{2\mu}}^{2\sqrt{2\mu}} B_{T,\mu}(0,p) \dd p=\int_{-2\sqrt{2\mu}}^{2\sqrt{2\mu}} F_{T,\mu}(p/2) \dd p=2 \int_{-\sqrt{2\mu}}^{\sqrt{2\mu}}  F_{T,\mu}(p) \dd p = \frac{4}{\sqrt{\mu}} \ln \frac{\mu}{T}+O(1)
\end{equation}
and the claim follows.

Part (\ref{glag}).
Recall that
\begin{equation}
\langle g \vert A_{T,\mu}^{\BR_+}-a_{T,\mu}  \vert g\rangle=\int_\BR \dd p\,  \widehat g (p)^2\left(A_{T,\mu}(p)- a_{T,\mu}\right)
-\int_\BR \dd p\,\widehat g(p)\int_\BR \dd q\  \frac{1}{4\pi} B_{T,\mu}(p,q)\widehat g(q).
\end{equation}
By Lemma~\ref{intBdq_max}, the first summand is negative and thus also $\langle g \vert A_{T,\mu}^{\BR_+}-a_{T,\mu}  \vert g\rangle<0$.
Moreover, using Lemma~\ref{intBdq_max} and $0<\widehat g(p)\leq 1$ we have
\begin{equation}\label{bound_gg}
\vert \langle g \vert A_{T,\mu}^{\BR_+}-a_{T,\mu}  \vert g\rangle \vert \leq \int_\BR \dd p\, \widehat g (p)^2 a_{T,\mu} 
+\int_\BR \dd p\, \widehat g(p) \int_\BR \dd q\  \frac{1}{4\pi} B_{T,\mu}(0,q).
\end{equation}
In both terms, the integral over $p$ gives a finite constant independent of $T$.
The claim follows from the asymptotics in Lemma~\ref{asymptotic_int_F}.
\end{proof}

\section{Weak Coupling Limit}\label{sec:weak}
In \cite{samoilenka_boundary_2020} it was observed by numerical and non-rigorous analytical computations that the effect of boundary superconductivity disappears in the weak coupling limit, in the sense that  $\frac{T_{c}^{\BR_+}(v)-T_{c}^\BR(v)}{T_{c}^\BR(v)} \to 0$ for $v\to 0$.
In this section we shall verify this claim. 

Recall that the bulk critical temperature $T_c^\BR(v)$ is the unique $T>0$ such that $a_{T,\mu}=\frac{1}{v}$.
For the system on the half-line, we have by continuity of $\inf \sigma(H_T^{\BR^+})$ in $T$
\begin{equation}\label{Tc_halfline}
T_c^{\BR_+}(v)=\min\{T \in [0,\infty) \vert \sup\sigma(A_{T,\mu}^{\BR_+})=v^{-1}\}.
\end{equation}
We want to invert this function and view $v$ as function of $T_c^{\BR_+}$.
We define $\mathfrak{v}(T):=(\sup \sigma(A_{T,\mu}^{\BR_+}))^{-1}$.
Note that $\mathfrak{v}\circ T_c^{\BR_+} =\mathrm{id}$ and for all  $T>0$ we have $T_c^{\BR_+}(\mathfrak{v}(T))\leq T$.
 
The claim can be reformulated in terms of the operator $A_{T,\mu}^{\BR_+}$ and $a_{T,\mu}$ in the following way.

\begin{lemma}\label{weak_lim_reformulation}
$ \lim_{v\to 0}\frac{T_{c}^{\BR_+}(v)-T_{c}^\BR(v)}{T_{c}^\BR(v)}  = 0   \Leftrightarrow \lim_{T\to 0} \inf \sigma(a_{T,\mu}\BI-A_{T,\mu}^{\BR_+} )=0$.
\end{lemma}

\begin{proof}
By definition, we have $\sup \sigma(A_{T,\mu}^{\BR_+})=\frac{1}{\mathfrak{v}(T)}=a_{T_c^\BR(\mathfrak{v}(T)),\mu}$.
Hence, 
\begin{equation}
\lim_{T\to 0} \inf \sigma(a_{T,\mu}\BI-A_{T,\mu}^{\BR_+} )=\lim_{T\to 0} (a_{T,\mu}-a_{T_c^\BR(\mathfrak{v}(T)),\mu})=\frac{1}{\pi \sqrt{\mu}} \lim_{T\to 0} \ln\left(\frac{T_c^\BR ( \mathfrak{v}(T))}{T}\right)
\end{equation}
where in the last equality we used  Lemma~\ref{asymptotic_int_F} and that $T\geq T_c^{\BR_+}(\mathfrak{v}(T))\geq T_c^\BR(\mathfrak{v}(T))\geq 0$ and thus $ \lim_{T\to 0} T_c^\BR(\mathfrak{v}(T))=0$.
Therefore, 
\begin{equation}
\lim_{T\to 0} \inf \sigma(a_{T,\mu}\BI-A_{T,\mu}^{\BR_+})=0 \Leftrightarrow \lim_{T\to 0} \frac{T-T_{c}^\BR(\mathfrak{v}(T))}{T_{c}^\BR(\mathfrak{v}(T))} = 0\,.
\end{equation} 
There exists a sequence $(T_n)$ such that $T_n \to 0$ as $n\to \infty$ and $T_c^{\BR_+}(\mathfrak{v}(T_n))=T_n$ for all $n$.
Therefore,
\begin{equation}
\lim_{T\to 0} \frac{T-T_{c}^\BR(\mathfrak{v}(T))}{T_{c}^\BR(\mathfrak{v}(T))} = \lim_{T\to 0} \frac{T_c^{\BR_+}(\mathfrak{v}(T))-T_{c}^\BR(\mathfrak{v}(T))}{T_{c}^\BR(\mathfrak{v}(T))} \,.
\end{equation}
Since $\lim_{T\to 0} T_c^\BR(\mathfrak{v}(T))=0$, also $\lim_{T\to 0} \mathfrak{v}(T)=0$.
Thus,
\begin{equation}
\lim_{T\to 0} \frac{T_c^{\BR_+}(\mathfrak{v}(T))-T_{c}^\BR(\mathfrak{v}(T))}{T_{c}^\BR(\mathfrak{v}(T))} =\lim_{v\to 0} \frac{T_c^{\BR_+}(v)-T_{c}^\BR(v)}{T_{c}^\BR(v)}
\end{equation}
and the claim follows.
\end{proof}

Recall the definition of $A_{T,\mu}$ in \eqref{intBdq}. With the notation 
\begin{equation}
E_{T,\mu}(p)=4 \pi \left(a_{T,\mu}-A_{T,\mu}(p)\right)
\end{equation}
we have for all $\psi \in L^2((0,\infty))$
\begin{equation}
4\pi (a_{T,\mu}\BI-A_{T,\mu}^{\BR_+}) \psi (p)= E_{T,\mu}(p) \psi(p)+\int_\BR B_{T,\mu}(p,q) \psi(q) \dd q.
\end{equation}

For the proof of Theorem~\ref{main_result} \eqref{thm_weak}, we need the following intermediate results which are proved in Section~\ref{sec:weaklea}.
\begin{lemma}\label{lea:supB}
Let $\mu>0$. Then
\begin{equation}
\sup_{T>0} \lVert B_{T,\mu}\rVert <\infty
\end{equation}
\end{lemma}

\begin{lemma}\label{lea:1B1}
Let $\BI_{\leq \epsilon}$ denote multiplication with the characteristic function of the interval $[-\epsilon,\epsilon]$ in momentum space.
Let $\mu>0$.
Then
\begin{equation}
\lim_{\epsilon \to 0} \sup_T \lVert \BI_{\leq \epsilon} B_{T,\mu} \BI_{\leq \epsilon} \lVert \leq \lim_{\epsilon \to 0} \sup_T \lVert \BI_{\leq \epsilon} B_{T,\mu} \BI_{\leq \epsilon} \lVert_{\rm{HS}}  =0,
\end{equation}
where $\lVert \cdot \rVert_{\rm{HS}}$ denotes the Hilbert--Schmidt norm.
\end{lemma}

\begin{lemma}\label{lea:Eln}
Let $0<\epsilon<2 \sqrt{\mu}$.
For $\vert p\vert>\epsilon$ 
we have 
\begin{equation}
E_{T,\mu}(p) \geq c_1  \ln \left( \frac{c_2}{T}\right)
\end{equation}
for constants $c_1, c_2>0$ and $T$ small enough.
\end{lemma}

\begin{proof}[Proof of Theorem~\ref{main_result} \eqref{thm_weak}]
By Lemma \ref{weak_lim_reformulation} it suffices to prove $0=\lim_{T\to 0} \inf \sigma(a_{T,\mu}\BI-A_{T,\mu}^{\BR_+} )=\lim_{T\to 0}\frac{1}{4\pi} \inf \sigma(E_{T,\mu}+B_{T,\mu})$.
By \eqref{T1>T0}, we only need to show that $\lim_{T\to0} \inf \sigma(E_{T,\mu}+B_{T,\mu}) \geq 0$.
For $\delta>0$ 
we can write
\begin{equation}\label{E+B+d}
E_{T,\mu}+B_{T,\mu}+\delta=\sqrt{E_{T,\mu}+\delta}\left(\BI + \frac{1}{\sqrt{E_{T,\mu}+\delta}}B_{T,\mu}\frac{1}{\sqrt{E_{T,\mu}+\delta}}\right)\sqrt{E_{T,\mu}+\delta}
\end{equation}
since $E_{T,\mu}(p)\geq 0$ by Lemma \ref{intBdq_max}.
We shall show that for all $\delta>0$ 
\begin{equation}\label{eq:EBEto0}
\lim_{T\to 0} \left\lVert \frac{1}{\sqrt{E_{T,\mu}+\delta}}B_{T,\mu}\frac{1}{\sqrt{E_{T,\mu}+\delta}}\right\rVert=0 \,.
\end{equation}
Hence, the operator in the bracket in \eqref{E+B+d} is positive for small $T$.
This implies that for all $\delta>0$ for $T$ small enough we have $\inf \sigma(E_{T,\mu}+B_{T,\mu}+\delta)>0$.
Since $\delta$ can be arbitrarily small, the theorem follows.

To prove \eqref{eq:EBEto0}, we use the notation of Lemma~\ref{lea:1B1} and estimate for an arbitrary $0<\epsilon<2\sqrt{\mu}$
\begin{multline}
\left\lVert \frac{1}{\sqrt{E_{T,\mu}+\delta}}B_{T,\mu}\frac{1}{\sqrt{E_{T,\mu}+\delta}}\right\rVert \leq 
\left\lVert \BI_{\leq\epsilon} \frac{1}{\sqrt{E_{T,\mu}+\delta}}B_{T,\mu}\frac{1}{\sqrt{E_{T,\mu}+\delta}}\BI_{\leq\epsilon}\right\rVert\\
+\left\lVert \BI_{\leq\epsilon} \frac{1}{\sqrt{E_{T,\mu}+\delta}}B_{T,\mu}\frac{1}{\sqrt{E_{T,\mu}+\delta}}\BI_{>\epsilon}\right\rVert
+\left\lVert \BI_{>\epsilon} \frac{1}{\sqrt{E_{T,\mu}+\delta}}B_{T,\mu}\frac{1}{\sqrt{E_{T,\mu}+\delta}}\right\rVert\,.
\end{multline}
Now we use that $E_{T,\mu}\geq 0$ and Lemma~\ref{lea:Eln} to obtain
\begin{equation}
\lim_{T\to 0}\left\lVert \frac{1}{\sqrt{E_{T,\mu}+\delta}}B_{T,\mu}\frac{1}{\sqrt{E_{T,\mu}+\delta}}\right\rVert \leq \lim_{T\to 0} \frac{1}{\delta}\left\lVert \BI_{\leq\epsilon} B_{T,\mu} \BI_{\leq\epsilon}\right\rVert+\lim_{T\to 0} \frac{2 c_1^{1/2}}{(\delta \vert \ln(c_2/T)\vert)^{1/2}} \lVert B_{T,\mu} \rVert.
\end{equation}
 With Lemma~\ref{lea:supB} it follows that the second term vanishes and 
\begin{equation}
\lim_{T\to 0}\left\lVert \frac{1}{\sqrt{E_{T,\mu}+\delta}}B_{T,\mu}\frac{1}{\sqrt{E_{T,\mu}+\delta}}\right\rVert \leq \sup_T \frac{1}{\delta}\left\lVert \BI_{\leq\epsilon} B_{T,\mu} \BI_{\leq\epsilon}\right\rVert.
\end{equation}
Since $\epsilon>0$ was arbitrary, \eqref{eq:EBEto0} follows from Lemma~\ref{lea:1B1}.
\end{proof}

\begin{rem}
In the case of Neumann boundary conditions, the same argument proves \eqref{thm_n_weak}.
\end{rem}

\subsection{Proofs of Intermediate Results}\label{sec:weaklea}

\begin{proof}[Proof of Lemma~\ref{lea:supB}]
In order to bound $B_{T,\mu}(p,q)$ we apply the following inequality proved in Section~\ref{sec:pfweak}.
\begin{lemma}\label{lea:estB1}
For all  $x,y \in \BR$ and $T>0$ it holds that 
\begin{equation}
\frac{\tanh(x/T)+\tanh(y/T)}{x+y} < \frac{2}{\vert x\vert +\vert y \vert}.
\end{equation}
\end{lemma}
Hence,  $B_{T,\mu}(p,q)$ is bounded above by 
\begin{equation}
f(p,q)= \frac{2}{\vert \left(\frac{p+q}{2}\right)^2-\mu \vert +\vert \left(\frac{p-q}{2}\right)^2-\mu \vert}.
\end{equation}
The function $f$ has singularities at the four points where $\{\vert p\vert ,\vert q\vert \}=\{0,2\sqrt{\mu}\}$.
Since $f$ diverges linearly at those points, the idea is to do a Schur test with a test function of the form $d(p)^\alpha$, where $d(p)$ is the distance from the singularities in variable $p$ and $\alpha \in (0,1)$.
We choose the function $h(p)=\min\{\vert p\vert , \vert 2\sqrt{\mu}-\vert p\vert \vert\}^{1/2}$.
The Schur test gives
\begin{equation}\label{schurtestb}
\sup_T \lVert B_{T,\mu} \rVert \leq  \sup_T \sup_ph(p) \int_\BR \frac{B_{T,\mu}(p,q)}{h(q)} \dd q \leq \sup_p  h(p) \int_\BR \frac{f(p,q)}{h(q)} \dd q =2 \sup_{p>0}  h(p) \int_0^\infty \frac{f(p,q)}{h(q)} \dd q,
\end{equation}
where we used that $\frac{h(p)f(p,q)}{h(q)}=\frac{h(\vert p \vert)f(\vert p\vert , \vert q \vert)}{h(\vert q \vert)}$ for the last equality.

In order to estimate $h(p)\int_0^\infty \frac{f(p,q)}{h(q)} \dd q$, we split the domain into nine regions as indicated in Figure~\ref{supB_domains}.
The finiteness of the right hand side of \eqref{schurtestb} follows from the bounds listed in Table~\ref{tab:1}.
In the following, we prove the bounds in Table~\ref{tab:1}.

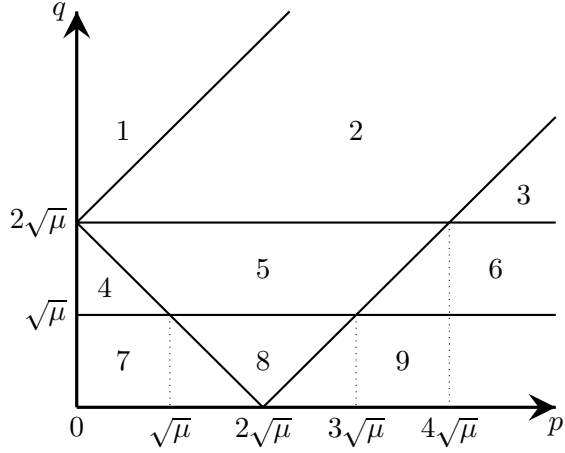
\begin{figure}
\centering
\begin{tikzpicture}[scale=0.35]
\tikzmath{\smu=3.5;\l1=18;\l2=15;} 

\draw (0,0) node[below]{$0$};
\draw (0,\smu) node[left]{$\sqrt{\mu}$};
\draw (0,2*\smu) node[left]{$2\sqrt{\mu}$};
\draw (\smu,0) node[below]{$\sqrt{\mu}$};
\draw (2*\smu,0) node[below]{$2\sqrt{\mu}$};
\draw (3*\smu,0) node[below]{$3\sqrt{\mu}$};
\draw (4*\smu,0) node[below]{$4\sqrt{\mu}$};
\draw (0,\l2) node[left]{$q$};
\draw (\l1,0) node[below]{$p$};

\draw (\smu/2,3*\smu) node{1};
\draw (3*\smu,3*\smu) node{2};
\draw (4.8*\smu,2.3*\smu) node{3};
\draw (0.3*\smu,1.3*\smu) node{4};
\draw (2*\smu,1.5*\smu) node{5};
\draw (4.5*\smu,1.5*\smu) node{6};
\draw (0.5*\smu,0.5*\smu) node{7};
\draw (2*\smu,0.5*\smu) node{8};
\draw (3.5*\smu,0.5*\smu) node{9};

 \draw[black, very thick, decoration={markings, mark=at position 1 with {\arrow[scale=2,>=stealth]{>}}},
        postaction={decorate}]
 (0,0)--(\l1,0);
 \draw[black, very thick, decoration={markings, mark=at position 1 with {\arrow[scale=2,>=stealth]{>}}},
        postaction={decorate}]
 (0,0)--(0,\l2);

 \draw[black, thick]
 (0,2*\smu)--(2*\smu,0)
 (0,\smu)--(\l1,\smu)
 (0,2*\smu)--(\l1,2*\smu)
 (0,2*\smu)--(\l2-2*\smu,\l2)
 (2*\smu,0)--(\l1,\l1-2*\smu);
 
 \draw[black, dotted]
(\smu,\smu)--(\smu,0)
(3*\smu,\smu)--(3*\smu,0)
(4*\smu,2*\smu)--(4*\smu,0);

\end{tikzpicture} 
\caption{The nine regions of the domain of $p,q$ in the proof of Lemma~\ref{lea:supB}.}
\label{supB_domains}
\end{figure}

\begin{table}[ht]
\begin{center}
\begin{tabular}{|c|c|c |c|}
\hline
 Region & Expression & Upper bound & Proof \\
    \hline
    1 & $ h(p)\int_{2\sqrt{\mu}+p}^\infty \frac{f(p,q)}{h(q)} \dd q $ & $\frac{2}{\sqrt{\mu}}$ & \eqref{reg1} \\ 
    2 & $ h(p)\int_{\max\{2 \sqrt{\mu},p-2\sqrt{\mu}\}}^{2\sqrt{\mu}+p} \frac{f(p,q)}{h(q)} \dd q $ & $\frac{2}{\sqrt{\mu}}$ & \eqref{reg2} \\ 
    3 & $ h(p)\int_{2 \sqrt{\mu}}^{p-2\sqrt{\mu}} \frac{f(p,q)}{h(q)} \dd q $ & $\frac{2}{\sqrt{\mu}}$ & \eqref{reg3} \\ 
    4 & $ h(p)\int_{\sqrt{\mu}}^{2\sqrt{\mu}-p} \frac{f(p,q)}{h(q)} \dd q $ & $\frac{8\cdot 2^{1/2}}{\sqrt{\mu}}$  & \eqref{reg4} \\ 
    5 & $ h(p)\int_{\max\{2 \sqrt{\mu}-p,\sqrt{\mu},p-2\sqrt{\mu}\}}^{2\sqrt{\mu}} \frac{f(p,q)}{h(q)} \dd q $ &$\frac{4}{\sqrt{\mu}}$  & \eqref{reg5} \\ 
    6 & $ h(p)\int_{\sqrt{\mu}}^{\min\{2\sqrt{\mu},p-2\sqrt{\mu}\}} \frac{f(p,q)}{h(q)} \dd q $ & $ \frac{2^{1/2} 8}{\sqrt{\mu}(3-3^{3/4})}$  & \eqref{reg6} \\ 
    7 & $ h(p)\int_{0}^{\min\{\sqrt{\mu},2\sqrt{\mu}-p\}} \frac{f(p,q)}{h(q)} \dd q $ & $\frac{4}{\sqrt{\mu}}$ & \eqref{reg7} \\ 
    8 & $ h(p)\int_{\vert 2 \sqrt{\mu}-p\vert}^{\sqrt{\mu}} \frac{f(p,q)}{h(q)} \dd q $ &$\frac{4}{\sqrt{\mu}}$ & \eqref{reg8} \\ 
    9 & $ h(p)\int_{0}^{\min \{\sqrt{\mu},p-2\sqrt{\mu} \}} \frac{f(p,q)}{h(q)} \dd q $ & $\frac{2}{\sqrt{\mu}}$  & \eqref{reg9} \\ 
    \hline
  \end{tabular}
\caption{Overview of the estimates used in the proof of Lemma~\ref{lea:supB}.}
\label{tab:1}
\end{center}
\end{table}

In region 1, we have 
\begin{multline}\label{reg1}
\int_{2\sqrt{\mu}+p}^\infty \frac{f(p,q)}{h(q)} \dd q = \int_{2\sqrt{\mu}+p}^\infty \frac{4}{p^2+q^2-4\mu} \frac{1}{(q-2\sqrt{\mu})^{1/2}} \dd q \leq  
\int_{2\sqrt{\mu}+p}^\infty \frac{4}{(q+2\sqrt{\mu})(q-2\sqrt{\mu})^{3/2}} \dd q\\
\leq  
\frac{1}{\sqrt{\mu}}\int_{2\sqrt{\mu}+p}^\infty \frac{1}{(q-2\sqrt{\mu})^{3/2}} \dd q = \frac{2}{\sqrt{\mu} p^{1/2}}.
\end{multline}

In region 2, we have 
\begin{multline}\label{reg2}
\int_{\max\{2 \sqrt{\mu},p-2\sqrt{\mu}\}}^{2\sqrt{\mu}+p} \frac{f(p,q)}{h(q)} \dd q = \int_{\max\{2 \sqrt{\mu},p-2\sqrt{\mu}\}}^{2\sqrt{\mu}+p} \frac{2}{pq (q-2 \sqrt{\mu})^{1/2}} \dd q\leq \frac{2}{p}\int_{2 \sqrt{\mu}}^{2\sqrt{\mu}+p} \frac{1}{q (q-2 \sqrt{\mu})^{1/2}}  \dd q \\
\leq \frac{1}{p\sqrt{\mu} }\int_{2 \sqrt{\mu}}^{2\sqrt{\mu}+p} \frac{1}{ (q-2 \sqrt{\mu})^{1/2}}  \dd q = \frac{2 }{\sqrt{ \mu} p^{1/2} }.
\end{multline}

In region 3, we have $p>4 \sqrt{\mu}$ and 
\begin{multline}\label{reg3}
\int_{2 \sqrt{\mu}}^{p-2\sqrt{\mu}} \frac{f(p,q)}{h(q)} \dd q = \int_{2 \sqrt{\mu}}^{p-2\sqrt{\mu}} \frac{4}{p^2+q^2-4\mu} \frac{1}{(q-2\sqrt{\mu})^{1/2}} \dd q\\
\leq
\frac{4}{p^2}\int_{2 \sqrt{\mu}}^{p-2\sqrt{\mu}} \frac{1}{(q-2\sqrt{\mu})^{1/2}} \dd q
=\frac{8}{p^2} (p-4\sqrt{\mu})^{1/2} \leq \frac{8}{p^{3/2}} \leq \frac{2}{\sqrt{\mu} p^{1/2}}
\end{multline}
where we used $p>4 \sqrt{\mu}$ in the last inequality.

In region 4, we have $p<\sqrt{\mu}$ and
\begin{multline}\label{reg4}
\int_{\sqrt{\mu}}^{2\sqrt{\mu}-p} \frac{f(p,q)}{h(q)} \dd q= \int_{\sqrt{\mu}}^{2\sqrt{\mu}-p}  \frac{4}{4\mu-p^2-q^2} \frac{1}{(2\sqrt{\mu}-q)^{1/2}} \dd q \\ = \int_{\sqrt{\mu}}^{2\sqrt{\mu}-p}  \frac{4}{(\sqrt{4\mu-p^2}+q)(\sqrt{4\mu-p^2}-q)} \frac{1}{(2\sqrt{\mu}-q)^{1/2}} \dd q \\
 \leq \frac{1}{\sqrt{\mu}}\int_{\sqrt{\mu}}^{2\sqrt{\mu}-p}  \frac{4}{(\sqrt{4\mu-p^2}-q)} \frac{1}{(2\sqrt{\mu}-q)^{1/2}} \dd q  \leq \frac{1}{\sqrt{\mu}}\int_{-\infty}^{2\sqrt{\mu}-p}  \frac{4}{(\sqrt{4\mu-p^2}-q)^{3/2}}\dd q \\
=\frac{8}{\sqrt{\mu}(\sqrt{4\mu-p^2}-2\sqrt{\mu}+p)^{1/2}}=\frac{8}{\sqrt{\mu}(2\sqrt{\mu}-p)^{1/4}\left[(2\sqrt{\mu}+p)^{1/2}-(2\sqrt{\mu}-p)^{1/2}\right]^{1/2}}\\
= \frac{8\left[(2\sqrt{\mu}+p)^{1/2}+(2\sqrt{\mu}-p)^{1/2}\right]^{1/2}}{\sqrt{\mu}(2\sqrt{\mu}-p)^{1/4}\left(2p\right)^{1/2}}\leq \frac{8\left(4 \mu^{1/4}\right)^{1/2}}{\sqrt{\mu}\mu^{1/8}\left(2p\right)^{1/2}}=\frac{8 \cdot 2^{1/2}}{\sqrt{\mu} p^{1/2}},
\end{multline}
where we used $p<\sqrt{\mu}$ in the last inequality.

In region 5, we have 
\begin{multline}\label{reg5}
\int_{\max\{2 \sqrt{\mu}-p,\sqrt{\mu},p-2\sqrt{\mu}\}}^{2\sqrt{\mu}} \frac{f(p,q)}{h(q)} \dd q= \int_{\max\{2 \sqrt{\mu}-p,\sqrt{\mu},p-2\sqrt{\mu}\}}^{2\sqrt{\mu}}   \frac{2}{pq} \frac{1}{(2\sqrt{\mu}-q)^{1/2}} \dd q \\
\leq  \frac{2}{p\sqrt{\mu}}\int_{\max\{2 \sqrt{\mu}-p,\sqrt{\mu},p-2\sqrt{\mu}\}}^{2\sqrt{\mu}}\frac{1}{(2\sqrt{\mu}-q)^{1/2}} \dd q= \frac{4}{p\sqrt{\mu}} \min\{p, \sqrt{\mu}, 4\sqrt{\mu}-p\}^{1/2}\leq \frac{4}{\sqrt{\mu}p^{1/2}}.
\end{multline}

In region 6, we have $p>3\sqrt{\mu}$ and
\begin{multline}\label{reg6}
\int_{\sqrt{\mu}}^{\min\{2\sqrt{\mu},p-2\sqrt{\mu}\}}  \frac{f(p,q)}{h(q)} \dd q= \int_{\sqrt{\mu}}^{\min\{2\sqrt{\mu},p-2\sqrt{\mu}\}}   \frac{4}{p^2+q^2-4\mu} \frac{1}{(2\sqrt{\mu}-q)^{1/2}} \dd q \\
\leq \frac{4}{p^2-3\mu}  \int_{0}^{2\sqrt{\mu}}   \frac{1}{(2\sqrt{\mu}-q)^{1/2}} \dd q =\frac{8}{p^2-3\mu} (2\sqrt{\mu})^{1/2}=\frac{8}{(p+\sqrt{3\mu})(p-\sqrt{3\mu})} (2\sqrt{\mu})^{1/2}\\
\leq \frac{8}{\sqrt{3\mu}(p^{1/2}-(3\mu)^{1/4})(p^{1/2}+(3\mu)^{1/4})} (2\sqrt{\mu})^{1/2}
\leq  \frac{2^{1/2} 8}{\sqrt{\mu}(3-3^{3/4})p^{1/2}} .
\end{multline}

In region 7, we have 
\begin{multline}\label{reg7}
\int_{0}^{\min\{\sqrt{\mu},2\sqrt{\mu}-p\}} \frac{f(p,q)}{h(q)} \dd q= \int_{0}^{\min\{\sqrt{\mu},2\sqrt{\mu}-p\}}  \frac{4}{4\mu-p^2-q^2} \frac{1}{q^{1/2}} \dd q\\
 \leq \frac{4}{4\mu-p^2-\min\{\sqrt{\mu},2\sqrt{\mu}-p\}^2} \int_{0}^{\min\{\sqrt{\mu},2\sqrt{\mu}-p\}}   \frac{1}{q^{1/2}} \dd q = \frac{8  \min\{\sqrt{\mu},2\sqrt{\mu}-p\}^{1/2}  }{4\mu-p^2-\min\{\sqrt{\mu},2\sqrt{\mu}-p\}^2}\\
= \left\{ \begin{array}{cr} \frac{8 \mu^{1/4}}{3\mu-p^2}& \mathrm{if}\  p< \sqrt{\mu} \\  \frac{4}{ p (2\sqrt{\mu}-p)^{1/2}} &\mathrm{if}\  p>\sqrt{\mu} \end{array}\right.
\leq \left\{ \begin{array}{cr} \frac{4 \mu^{1/4}}{\mu}& \mathrm{if}\  p< \sqrt{\mu} \\  \frac{4}{\sqrt{\mu} (2\sqrt{\mu}-p)^{1/2}} &\mathrm{if}\  p>\sqrt{\mu} \end{array}\right.
\end{multline}

In region 8, we have $p>\sqrt{\mu}$ and
\begin{multline}\label{reg8}
\int_{\vert 2 \sqrt{\mu}-p\vert}^{\sqrt{\mu}}  \frac{f(p,q)}{h(q)} \dd q= \int_{\vert 2 \sqrt{\mu}-p\vert}^{\sqrt{\mu}}  \frac{2}{pq} \frac{1}{q^{1/2}} \dd q \leq \frac{2}{\sqrt{\mu}}\int_{\vert 2 \sqrt{\mu}-p\vert}^{\infty}  \frac{1}{q^{3/2}} \dd q = \frac{4}{\sqrt{\mu}} \vert 2 \sqrt{\mu}-p\vert^{-1/2}.
\end{multline}

In region 9, we have $p>2\sqrt{\mu}$ and 
\begin{multline}\label{reg9}
\int_{0}^{\min\{\sqrt{\mu},p-2\sqrt{\mu}\}}\frac{f(p,q)}{h(q)} \dd q= \int_{0}^{\min\{\sqrt{\mu},p-2\sqrt{\mu}\}} \frac{4}{p^2+q^2-4\mu} \frac{1}{q^{1/2}} \dd q\\
 \leq  \frac{4}{p^2-4\mu}\int_{0}^{\min\{\sqrt{\mu},p-2\sqrt{\mu}\}} \frac{1}{q^{1/2}} \dd q= \frac{8}{(p+2\sqrt{\mu})(p-2\sqrt{\mu})}\min\{\mu^{1/4}, (p-2\sqrt{\mu})^{1/2}\}\\
\leq \frac{2}{\sqrt{\mu} (p-2\sqrt{\mu})^{1/2}}
\end{multline}
\end{proof}

\begin{proof}[Proof of Lemma~\ref{lea:1B1}]
Let $0<\epsilon<\sqrt{\mu}$.
For $0\leq \vert p\vert ,\vert q\vert\leq\epsilon$ we have $2\mu-\left(\frac{p+q}{2}\right)^2-\left(\frac{p-q}{2}\right)^2 \geq 2\mu-2\epsilon^2$.
Together with $0\leq \tanh (x)\leq 1$ for $x\geq0$ we obtain
\begin{equation}
0\leq B_{T,\mu}(p,q)\leq \frac{1}{\mu-\epsilon^2}.
\end{equation}
Using this estimate, we bound the Hilbert--Schmidt norm as
\begin{equation}
\lVert \BI_{\leq \epsilon} B_{T,\mu} \BI_{\leq \epsilon} \lVert_\text{HS}^2=\int_{-\epsilon}^\epsilon  \int_{-\epsilon}^\epsilon B_{T,\mu}(p,q)^2 \dd p  \dd q \leq \frac{4\epsilon^2}{(\mu-\epsilon^2)^2}.
\end{equation}
\end{proof}

\begin{proof}[Proof of Lemma~\ref{lea:Eln}]
Recall that $E_{T,\mu}(p)=4\pi a_{T,\mu}-\int_\BR B_{T,\mu}(p,q) \dd q$.
The idea is to show that the supremum $\sup_{p>\epsilon,T>0} \int_\BR B_{T,\mu}(p,q) \dd q <\infty$.
Then, for $T\to 0$ we have $\inf_{\vert p\vert >\epsilon} E_{T,\mu}(p) \sim 4\pi a_{T,\mu} \sim \frac{4}{\sqrt{\mu}} \ln \frac{\mu}{T}$.

We shall prove that the following four expressions are finite.
\begin{equation} I_1:= \sup_{p>\epsilon,T>0} \int_{p+2\sqrt{\mu}}^\infty B_{T,\mu}(p,q) \dd q \end{equation}
\begin{equation} I_2:=\sup_{2\sqrt{\mu}>p>\epsilon,T>0} \int_0^{2\sqrt{\mu}-p}B_{T,\mu}(p,q) \dd q \end{equation}
\begin{equation}I_3:= \sup_{p>2\sqrt{\mu},T>0} \int_0^{p-2\sqrt{\mu}} B_{T,\mu}(p,q) \dd q \end{equation} 
\begin{equation} I_4:= \sup_{p>\epsilon,T>0} \int_{\vert p-2\sqrt{\mu} \vert}^{p+2\sqrt{\mu}} B_{T,\mu}(p,q) \dd q \end{equation}
From this, together with $B_{T,\mu}(p,q)=B_{T,\mu}(\vert p\vert ,\vert q\vert)$ it follows that
\begin{multline}
\sup_{\vert p\vert>\epsilon,T>0} \int_\BR B_{T,\mu}(p,q) \dd q \leq 2 \max\left\{\sup_{2\sqrt{\mu}> p >\epsilon,T>0} \int_0^\infty B_{T,\mu}(p,q) \dd q,\sup_{ p >2\sqrt{\mu},T>0}\int_0^\infty B_{T,\mu}(p,q) \dd q\right\} \\
\leq 2\max\{I_2+I_4+I_1, I_3+I_4+I_1\}<\infty.
\end{multline}

The following inequality is proved in Section~\ref{sec:pfweak}.
\begin{lemma}\label{lea:estB2}
For $x,y>0$
\begin{equation}
\frac{\tanh(x)-\tanh(y)}{x-y} \leq 4 e^{-2\min\{x,y\}}
\end{equation}
\end{lemma}
Applying Lemmas~\ref{lea:estB1} and \ref{lea:estB2} we estimate 
\begin{equation}\label{Eln_estB}
B_{T,\mu}(p,q)\leq \left\{\!\!\begin{array}{cr} \frac{2}{T} \exp\left(-{\min\left\{ (p+q)^2-4\mu , 4\mu-(p-q)^2 \right\}}/{4T}\right) & \text{for $\vert p-2\sqrt{\mu} \vert <q< p+2 \sqrt{\mu}$} \\
 \frac{8}{\lvert (p+q)^2-4\mu\rvert+\lvert (p-q)^2-4\mu\rvert } &\text{otherwise.} 
\end{array} \right.
\end{equation}
With \eqref{Eln_estB} we have
\begin{equation}
I_1 \leq \sup_{p>\epsilon} \int_{p+2\sqrt{\mu}}^\infty  \frac{4}{ p^2+q^2-4\mu} \dd q \leq \int_{\epsilon+2\sqrt{\mu}}^\infty  \frac{4}{ q^2-4\mu} \dd q <\infty\,.
\end{equation}
Furthermore,
\begin{multline}
I_2 \leq \sup_{2\sqrt{\mu}>p>\epsilon} \int_0^{2\sqrt{\mu}-p}  \frac{4}{ 4\mu-p^2-q^2} \dd q \leq \sup_{2\sqrt{\mu}>p>\epsilon} (2\sqrt{\mu}-p) \frac{4}{4\mu-p^2-(2\sqrt{\mu}-p)^2}\\
= \sup_{2\sqrt{\mu}>p>\epsilon} \frac{2}{p} = \frac{2}{\epsilon}
\end{multline}
Moreover,
\begin{equation}
I_3 \leq \sup_{p>2\sqrt{\mu}} \int_0^{p-2\sqrt{\mu}}  \frac{4}{ p^2+q^2-4\mu}  \dd q =\sup_{p>2\sqrt{\mu}}\frac{ 4 \arctan \left(\sqrt{\frac{p-2\sqrt{\mu}}{p+2\sqrt{\mu}}}\right)}{\sqrt{p^2-4\mu}} \leq \sup_{0<x<1}\frac{ \arctan \left(x\right)}{\sqrt{\mu} x}\leq \frac{1}{\sqrt{\mu}}
\end{equation}

In order to estimate $I_4$, note that $(p+q)^2-4\mu < 4\mu-(p-q)^2  \Leftrightarrow q<\sqrt{4\mu-p^2}$.
Let 
\begin{equation}
I_5:=\sup_{\epsilon<p<2\sqrt{\mu}, T>0} \frac{2}{T} \int_{2\sqrt{\mu}-p}^{\sqrt{4\mu-p^2}} e^{\mu/T-(p+q)^2/4T} \dd q,
\end{equation}
\begin{equation}
I_6:=\sup_{\epsilon<p<2\sqrt{\mu}, T>0} \frac{2}{T} \int_{\sqrt{4\mu-p^2}}^{2\sqrt{\mu}+p}e^{(q-p)^2/4T-\mu/T} \dd q,
\end{equation}
and
\begin{equation}
I_7:=\sup_{p>2\sqrt{\mu}, T>0} \frac{2}{T} \int_{p-2\sqrt{\mu}}^{p+2\sqrt{\mu}} e^{(q-p)^2/4T-\mu/T} \dd q.
\end{equation}
Then we have $I_4 \leq \max\{I_5+I_6, I_7\} $.
We can bound both $I_6$ and $I_7$ using
\begin{multline}
I_6, I_7 \leq \sup_{p>\epsilon, T>0} \frac{2}{T} \int_{p-2\sqrt{\mu}}^{p+2\sqrt{\mu}} e^{(q-p)^2/4T-\mu/T} \dd q = \sup_{ T>0} \frac{2}{T} \int_{-2\sqrt{\mu}}^{2\sqrt{\mu}} e^{q^2/4T-\mu/T} \dd q\\
 =\sup_{ T>0}  \frac{4 \sqrt{\pi} e^{-\mu/T}}{\sqrt{T}} \mathrm{erfi}\left(\sqrt{\frac{\mu}{T}}\right)=\frac{4\sqrt{\pi}}{\sqrt{\mu}}\sup_{x>0} x e^{-x^2}\mathrm{erfi}(x).
\end{multline}
Since $\sqrt{\pi} \lim_{x\to \infty} x e^{-x^2}\mathrm{erfi}(x)=1$, it follows that $I_6,I_7<\infty$.

Finally,
\begin{multline}
I_5 =\sup_{\epsilon<p<2\sqrt{\mu}, T>0} \frac{2 e^{\mu/T}}{T}\int_{2\sqrt{\mu}}^{\sqrt{4\mu-p^2}+p} e^{-q^2/4T} \dd q\leq \sup_{ T>0} \frac{2 e^{\mu/T}}{T}\int_{2\sqrt{\mu}}^{\infty} e^{-q^2/4T} \dd q\\
=\sup_{ T>0}\frac{2 \sqrt{\pi} e^{\mu/T}}{\sqrt{T}}\mathrm{erfc}\left(\sqrt{\frac{\mu}{T}}\right)=\frac{2 \sqrt{\pi}}{\sqrt{\mu}}\sup_{x>0} x e^{x^2} \mathrm{erfc}(x)
\end{multline}
Since $0 \leq \mathrm{erfc}(x)\leq 1$ and for $x\to \infty$ asymptotically $\mathrm{erfc}(x)\sim \frac{e^{-x^2}}{x \sqrt{\pi}}+o(e^{-x^2}/x)$, we have $\sup_{x>0} x e^{x^2} \mathrm{erfc}(x)<\infty$ and obtain $I_5 <\infty$.
\end{proof}

\section{Strong Coupling Limit}\label{sec:strong}
The goal of this section is to prove part \eqref{thm_strong} of Theorem~\ref{main_result}.
As for the weak coupling limit, we first translate the question about the relative temperature difference into a condition on $A_{T,\mu}^{\BR_+}$ and $a_{T,\mu}$.
While the weak coupling limit turned out to be equivalent to a low temperature limit, the strong coupling limit corresponds to a high temperature limit.
In this limit, the relevant quantities behave as follows.
\begin{lemma}\label{lea:limits}
Let $\mu>0$. Then
\begin{enumerate}[(i)]
\item $\lim_{v\to \infty} T_c^{\BR_+}(v)=\infty$ \label{lim2}
\item $\lim_{T\to \infty} T_c^\BR(\mathfrak{v}(T))=\infty$ \label{lim1}
\item $\lim_{T\to \infty} T^{1/2} a_{T,\mu}=a_{1,0}$ \label{lim3}
\item $\lim_{T\to \infty} T^{1/2} \sup \sigma(A_{T,\mu}^{\BR_+})=\sup \sigma(A_{1,0}^{\BR_+})$ \label{lim4}
\end{enumerate}
\end{lemma}
The proof is provided in Section~\ref{sec:stronglea}.
We can reformulate Theorem~\ref{main_result}\eqref{thm_strong} as follows.
\begin{lemma}\label{lea:thm3}
\begin{equation}
\lim_{v\to \infty} \frac{T_c^{\BR_+}(v)-T_c^\BR(v)}{T_c^\BR(v)}=0 \Leftrightarrow \sup \sigma(A_{1,0}^{\BR_+}) = a_{1,0}
\end{equation}
\end{lemma}
\begin{proof}
By Lemma~\ref{lea:limits}\eqref{lim4} and the definition of $\mathfrak{v}(T)$ we have 
\begin{equation}
\sup \sigma (A_{1,0}^{\BR_+})=\lim_{T\to \infty} T^{1/2} \sup \sigma(A_{T,\mu}^{\BR_+})=\lim_{T\to \infty} T^{1/2} a_{T_c^\BR(\mathfrak{v}(T)),\mu}
\end{equation}
By Lemma~\ref{lea:limits}\eqref{lim1} and \eqref{lim3} we get
\begin{equation}
\lim_{T\to \infty} T^{1/2} a_{T_c^\BR(\mathfrak{v}(T)),\mu}=a_{1,0}\lim_{T\to \infty}\left( \frac{T}{T_c^\BR(\mathfrak{v}(T))}\right)^{1/2}=a_{1,0}\lim_{v\to \infty}\left( \frac{T_c^{\BR_+}(v)}{T_c^\BR(v)}\right)^{1/2}
\end{equation}
where we used Lemma~\ref{lea:limits}\eqref{lim2} and $\mathfrak{v}(T_c^{\BR_+}(v))=v$ for the second equality.
Since $a_{1,0}>0$, the claim follows.
\end{proof}

\begin{rem}
In the case of Neumann boundary conditions, $d:=\sup \sigma(A_{1,0}^{\BR_+})-a_{1,0}>0$.
With the argument in Lemma~\ref{lea:thm3}, we have 
\begin{equation}
\lim_{v\to \infty} \frac{T_c^{\BR_+}(v)-T_c^\BR(v)}{T_c^\BR(v)}=\left(\frac{d}{a_{1,0}}+1\right)^2-1>0\,.
\end{equation}
\end{rem}

We are thus left with showing that $\sup \sigma(A_{1,0}^{\BR_+})=a_{1,0}$.
Recall that $\sup \sigma_{\text{ess}}(A_{1,0}^{\BR_+})=a_{1,0}$.
Hence it suffices  to prove that for all $\psi \in L^2((0,\infty))$
\begin{equation}\label{highTlimineq}
\langle \psi \vert A_{1,0}^{\BR_+} |\psi \rangle = \frac{1}{8\pi} \int_\BR \int_\BR B_{1,0}(p,q)\vert \psi(p)-\psi(q)\vert^2 \dd p \dd q \leq \frac{1}{4\pi} \int_\BR \vert \psi(p)\vert ^2 \dd p \int_\BR B_{1,0}(0,q) \dd q =\lVert \psi \rVert_2^2 a_{1,0}.
\end{equation}
In order to show this, we shall bound $B_{1,0}$ by a positive definite kernel $K$, in such a way that the right hand side of \eqref{highTlimineq} does not change.
\begin{lemma}\label{highTlim_K}
Let $K$ be the operator on $L^2(\BR^2) $ with integral kernel
\begin{equation}
K(p,q)= \min\{B_{1,0}(p,0),B_{1,0}(q,0)\}
\end{equation}
Then $K$ satisfies
\begin{enumerate}[(i)]
\item $B_{1,0}(p,q) \leq K(p,q)$ for all $p,q \in \BR$ \label{i1}
\item $K(p,q)=K(q,p)$ for all $p,q \in \BR$ \label{i3}
\item $K$ is positive definite \label{i4}
\item $\int_\BR K(p,q) \dd q \leq \int_\BR K(0,q)\dd q $ for all $p\in \BR$\label{i5}
\item $B_{1,0}(p,0)=K(p,0)$ for all $p\in \BR$ \label{i2}
\end{enumerate}
\end{lemma}
This implies \eqref{highTlimineq} and hence part \eqref{thm_strong} of Theorem~\ref{main_result} since
\begin{align}
\frac{1}{2} \int_\BR \int_\BR B_{1,0}(p,q)\vert \psi(p)-\psi(q)\vert^2 \dd p \dd q &\leq \frac{1}{2} \int_\BR \int_\BR K(p,q)\vert\psi(p)-\psi(q)\vert^2 \dd p \dd q \nonumber \\
&= \int_\BR \vert \psi(p)\vert ^2 \int_\BR K(p,q)\dd q\dd p - \langle \psi \vert K| \psi \rangle \nonumber \\
&\leq \int_\BR \vert \psi(p)\vert ^2 \int_\BR K(p,q)\dd q\dd p \nonumber \\
&\leq \lVert \psi \rVert_2^2\int_\BR K(0,q)\dd q=\lVert \psi \rVert_2^2\int_\BR B_{1,0}(0,q)\dd q \,.
\end{align}

\begin{proof}[Proof of Lemma~\ref{highTlim_K}]
Property \eqref{i3} is obvious.
Properties \eqref{i5} and \eqref{i2} follow from the fact that 
\begin{equation}
K(p,q)=\min\{F_{1,0}(p/2),F_{1,0}(q/2)\}=F_{1,0}(\max\{\vert p\vert ,\vert q\vert\}/2),
\end{equation}
where $F_{1,0}(p)=\frac{\tanh(p^2/2)}{p^2}$ has a maximum at $p=0$ and is monotonously decreasing for $p>0$.
For \eqref{i1} consider the following inequality, which is proved in Section~\ref{sec:pfstrong}.
\begin{lemma}\label{highTlimineq2}
For all $p,q \in \BR$ 
\begin{equation}
B_{1,0}(p,q)\leq \frac{\tanh\left(\frac{p^2+q^2}{8}\right)}{\frac{p^2+q^2}{4}}
\end{equation}
\end{lemma}
Together with the monotonicity of $\tanh(p)/p$ for $p\geq0$, it implies \eqref{i1}.
For property \eqref{i4} it suffices to show that there is a real-valued function $g$ such that 
\begin{equation}\label{prod}
K(p,q)=\int_\BR g(r,p)g(r,q) \dd r.
\end{equation}
In fact, let $g(r,p)=\sqrt{h(r) }\chi_{r>p^2}$ with
\begin{equation}
h(r)=\left.\frac{\dd}{\dd x} \frac{\tanh(x/2)}{x}\right\vert_{x=-r} \geq 0.
\end{equation}
With this choice, \eqref{prod} holds since
\begin{multline}
\int_\BR g(r,p)g(r,q) \dd r=\int_{\max\{p^2,q^2\}}^\infty h(r)\dd r= \int_{-\infty}^{-\max\{p^2,q^2\}}\frac{\dd}{\dd x} \left. \frac{\tanh(x/2)}{x}\right\vert_{x=r} \dd r\\
=  \frac{\tanh(\max\{p^2,q^2\}/2)}{\max\{p^2,q^2\}}=K(p,q)
\end{multline}
\end{proof}

\subsection{Proof of  Lemma~\ref{lea:limits}}\label{sec:stronglea}
\begin{proof}[Proof of Lemma~\ref{lea:limits}]
For \eqref{lim2} we have $\lim_{v\to \infty} T_c^{\BR_+}(v) \geq \lim_{v\to \infty} T_c^\BR(v)= \infty$ by \eqref{T1>T0}.

\eqref{lim1} follows easily from \eqref{lim4}: 
Clearly \eqref{lim4} implies that 
\begin{equation}
\lim_{T\to \infty} \sup \sigma(A_{T,\mu}^{\BR_+})=0.
\end{equation}
Since $a_{T_c^\BR(\mathfrak{v}(T)),\mu}=\sup \sigma (A_{T_,\mu}^{\BR_+})$ this is equivalent to 
\begin{equation}
\lim_{T\to \infty}a_{T_c^\BR(\mathfrak{v}(T)),\mu}=0.
\end{equation}
Using that $a_{T,\mu}$ is strictly decreasing in $T$ with $\lim_{T\to \infty} a_{T,\mu}=0$, this in turn is equivalent to
\begin{equation}
\lim_{T\to \infty}T_c^\BR(\mathfrak{v}(T))=\infty.
\end{equation}

For \eqref{lim3} we have after substituting $q/2T^{1/2} \to q$
\begin{equation}
\lim_{T\to \infty} T^{1/2} a_{T,\mu}= \frac{1}{2\pi} \lim_{T\to \infty} \int_\BR \frac{\tanh\left(\frac{q^2-\mu/T}{2}\right)}{q^2-\mu/T} \dd q.
\end{equation}
Fix some $T_0>0$.
Since $\tanh(x)/x$ is decreasing for $x\geq 0$ and bounded by $1$, the integrand is bounded by $\frac{1}{2}\chi_{\vert q\vert <2\sqrt{\mu/T_0}}+\frac{1}{q^2-\mu/T_0}\chi_{\vert q\vert >2\sqrt{\mu/T_0}}$ for $T>T_0$.
This is an $L^1$ function, so by dominated convergence we can pull the limit into the integral and arrive at the claim.

\eqref{lim4} 
Let $U_T$ denote the unitary transformation $U_T \psi(p)= T^{1/4} \psi(T^{1/2} p)$ on $L^2(\BR^2)$.
We shall prove that $\lim_{T\to \infty} \lVert U_T T^{1/2} A_{T,\mu}^{\BR_+} U_T^\dagger- A_{1,0}^{\BR_+}\rVert=0$, which implies the claim.
Note that 
\begin{equation}
 U_T T^{1/2} A_{T,\mu}^{\BR_+} U_T^\dagger = A_{1,\mu/T}^{\BR_+} 
\end{equation}
Therefore, we have
\begin{multline}\label{high_T_lim_sum}
 \lim_{T\to \infty} \lVert U_T T^{1/2} A_{T,\mu}^{\BR_+} U_T^\dagger- A_{1,0}^{\BR_+}\rVert= \lim_{\mu\to 0} \lVert A_{1,\mu}^{\BR_+}- A_{1,0}^{\BR_+}\rVert\\
 \leq \frac{1}{4\pi}\lim_{\mu \to 0 }  \sup_p  \left \vert \int_\BR ( B_{1,\mu}(p,q)-B_{1,0}(p,q)) \dd q\right \vert +\frac{1}{4\pi}\lim_{\mu \to 0} \lVert B_{1,\mu}-B_{1,0} \rVert
\end{multline}
For the second term on the second line of \eqref{high_T_lim_sum} we bound the operator norm by the Hilbert--Schmidt norm
\begin{equation}
\lVert B_{1,\mu}-B_{1,0} \rVert^2 \leq \lVert B_{1,\mu}-B_{1,0} \rVert_{\text{HS}}^2 = \int_\BR \dd p \int_\BR \dd q \left(B_{1,\mu}(p,q)-B_{1,0}(p,q)\right)^2
\end{equation}
Using that $B_{T,\mu}(p,q)\leq 1/2T$ and $\vert \tanh(x) \vert \leq 1 $ one can bound
\begin{equation}
B_{1,\mu}(p,q)^2\leq \frac{1}{4}\chi_{p^2+q^2\leq 4 \mu}(p,q)+\chi_{p^2+q^2> 4 \mu}(p,q) \min \left\{ \frac{1}{4}, \frac{16}{(p^2+q^2-4\mu)^2}\right\}=: f_\mu (p,q).
\end{equation}
By the monotonicity of $f_\mu $ in $\mu$, we have for all $\nu \leq \mu$ that $\left(B_{1,\nu}(p,q)-B_{1,0}(p,q)\right)^2\leq 2 f_\mu (p,q)$.
Since $f_\mu$ is an $L^1$ function, dominated convergence implies $\lim_{\mu \to 0} \lVert B_{1,\mu}-B_{1,0} \rVert=0$.

For the first term in the second line of \eqref{high_T_lim_sum} we estimate
\begin{multline}\label{high_T_lim2}
 \lim_{\mu \to 0 }  \sup_p  \left \vert \int_\BR ( B_{1,\mu}(p,q)-B_{1,0}(p,q)) \dd q\right \vert = \lim_{\mu \to 0 }  \sup_p  \left \vert \int_\BR \int_0^\mu \frac{\p}{\p \nu} B_{1,\nu}(p,q) \dd \nu \dd q\right \vert \\
\leq \lim_{\mu \to 0 }   \mu \sup_p  \sup_{\nu \in [0,\mu]}\int_\BR \left \vert\frac{\p}{\p \nu} B_{1,\nu}(p,q)\right \vert \dd q,
\end{multline}
where we used the triangle inequality and Fubini's theorem in the last step.
By \eqref{L_series} we may write
\begin{multline}\label{b_der_series}
\frac{\p}{\p \mu} B_{1,\mu}(p,q)= 2 \sum_{n\in \BZ} \frac{1}{\left(\left(\frac{p+q}{2}\right)^2-\mu-i w_n \right)^2} \frac{1}{\left(\frac{p-q}{2}\right)^2-\mu+i w_n }\\
+ \frac{1}{\left(\frac{p+q}{2}\right)^2-\mu-i w_n } \frac{1}{\left(\left(\frac{p-q}{2}\right)^2-\mu+i w_n \right)^2},
\end{multline}
where $w_n=\pi(2n+1)$.
Observe that 
\begin{equation}
\left \vert \left(\frac{p+q}{2}\right)^2-\mu-i w_n \right \vert \geq w_n \chi_{\vert q\vert<2\sqrt{\mu}}+\sqrt{\left(q^2/4-\mu\right)^2+w_n^2} \chi_{\vert q\vert >2\sqrt{\mu}}
\end{equation}
and
\begin{equation}
\left \vert \left(\frac{p-q}{2}\right)^2-\mu+i w_n \right \vert \geq w_n .
\end{equation}
Applying Fubini's theorem to swap integration and summation, we have for all $p$ and $\mu$
\begin{multline}
\int_\BR \left \vert\frac{\p}{\p \mu} B_{1,\mu}(p,q)\right \vert \dd q \leq 2 \sum_{n\in \BZ} \int_\BR \dd q \Bigg(\frac{2}{w_n^3}\chi_{\vert q\vert\leq2\sqrt{\mu}}+\frac{\chi_{\vert q\vert>2\sqrt{\mu}}}{w_n \left((q^2/4-\mu)^2+w_n^2\right)}\\
 + \frac{\chi_{\vert q\vert>2\sqrt{\mu}}}{w_n^2 \sqrt{(q^2/4-\mu)^2+w_n^2}}\Bigg)\\
=2 \sum_{n\in \BZ}\left[\frac{8\sqrt{\mu}}{w_n^3}+ \frac{2}{w_n^{5/2}}\int_0^\infty \dd s \left(\frac{1}{\left(s^2+1\right)\sqrt{s+\mu/w_n}}+ \frac{1}{\sqrt{(s^2+1)(s+\mu/w_n)} }\right)\right],
\end{multline}
where we substituted $s=w_n^{-1}(q^2/4-\mu)$.
For $\mu<1$ we therefore obtain a $\mu$-independent bound
\begin{multline}
\sup_p \sup_{\nu \in [0,\mu]}\int_\BR \left \vert\frac{\p}{\p \nu} B_{1,\nu}(p,q)\right \vert \dd q \\
\leq 2 \sum_{n\in \BZ}\left[\frac{8}{w_n^3}+ \frac{2}{w_n^{5/2}}\int_0^\infty \dd s \left(\frac{1}{\left(s^2+1\right)\sqrt{s}}+ \frac{1}{\sqrt{(s^2+1)s} }\right)\right]<\infty.
\end{multline}
Thus, the last expression in \eqref{high_T_lim2} vanishes and the claim follows.
\end{proof}

\section{Proofs of Auxiliary Results}\label{sec:pf}
\subsection{From Section~\ref{sec:pre}}\label{sec:pfpre}
\begin{proof}[Proof of Lemma~\ref{lea:L_laplace}]
Note that for all $p,q\in \BR$
\begin{equation}
L_{T,\mu}(p,q)\leq \min\left\{\frac{1}{2T}, \frac{2}{\vert p^2+q^2-2\mu \vert}\right\}
\end{equation}
Hence, $L_{T,\mu}(p,q) (1+p^2+q^2) \leq \frac{1+4T+2\mu}{2T}$ and $L_{T,\mu}(p,q) (T+p^2+q^2) \leq \frac{5T+2\mu}{2T}$.
So with $C_1(T,\mu)=\frac{2T}{1+4T+2\mu}$ and $C_3(T_0,\mu)=\frac{2T_0}{5T_0+2\mu}$ the respective inequalities hold.

For the remaining inequality, note that $L_{T,\mu}$ vanishes only at infinity.
Let $\epsilon>0$. 
There is a constant $c_1$ such that $L_{T,\mu}(p,q)>c_1$ for all $\vert p \vert, \vert q \vert \leq \sqrt{\max\{2\mu,0\}+\epsilon}$.
Moreover, if $\vert p \vert$ or $\vert q \vert > \sqrt{\max\{2\mu,0\}+\epsilon}$, we have 
\begin{equation}
L_{T,\mu}(p,q)\geq \frac{\tanh((\vert \mu \vert+\epsilon)/2T)-\tanh(\mu/2T)}{p^2+q^2-2\mu}\geq \frac{c_2}{p^2+q^2+\max\{-2\mu,0\}}
\end{equation}
In particular, $L_{T,\mu}(p,q) (1+p^2+q^2)\geq \min\{c_1,c_2, c_2/\max\{-2\mu,0\}\}$.
\end{proof}

\begin{proof}[Proof of Lemma~\ref{intBdq_max}]
First, we show that for every $x,y\in \BR$
\begin{equation}\label{L_leq_F+F}
\frac{\tanh(x)+\tanh(y)}{x+y}\leq \frac{1}{2} \left(\frac{\tanh(x)}{x}+\frac{\tanh(y)}{y}\right)
\end{equation}
Since changing $x\to -x, y \to -y$ does not change the expressions, we may assume without loss of generality that $x\geq \vert y \vert$.
Note that
\begin{equation}
\frac{\tanh(x)+\tanh(y)}{x+y}= \frac{1}{2(x+y)} \left[(x+y)\left(\frac{\tanh(x)}{x}+\frac{\tanh(y)}{y}\right)+(x-y) \left(\frac{\tanh(x)}{x}-\frac{\tanh(y)}{y}\right)\right]
\end{equation}
Since $\tanh(x)/x\leq\tanh(y)/y$, the last term is not positive and the inequality \eqref{L_leq_F+F} follows.

For $p\in \BR$ we therefore have 
\begin{equation}\label{E_est_F}
\int_\BR  B_{T,\mu}(p,q) \dd q \leq \frac{1}{2}\int_\BR \left[ F_{T,\mu}\left(\frac{p+q}{2}\right) +F_{T,\mu}\left(\frac{p-q}{2}\right)\right] \dd q =\int_\BR F_{T,\mu}(q/2) \dd q.
\end{equation}
Since $F_{T,\mu}(q/2)=B(0,q)$, the claim follows.
\end{proof}
\subsection{From Section~\ref{sec:ex}}\label{sec:pfex}
\begin{proof}[Proof of Lemma~\ref{asymptotic_int_F}]
Substituting by $p^2-\mu=t$ for $p^2>\mu$ and $\mu-p^2=t$ for $p^2<\mu$ we get
\begin{equation}
\int_\BR F_{T,\mu}(p) \dd p =2\int_0^\infty \frac{\tanh\left(\frac{p^2-\mu}{2T}\right)}{p^2-\mu}\dd p= 2\int_0^\infty \frac{\tanh(t/2T)}{2t \sqrt{\mu+t}} \dd t+2\int_0^\mu \frac{\tanh(t/2T)}{2t \sqrt{\mu-t}} \dd t.
\end{equation}
It was shown in \cite[Lemma 1]{hainzl_critical_2008} that 
\begin{equation}
\lim_{T\to0} \left( \int_0^\mu \frac{\tanh(t/2T)}{t} \dd t - \ln \frac{\mu}{T} \right)= \gamma- \ln \frac{\pi}{2}.
\end{equation}
By monotone convergence, we observe that
\begin{equation}
\lim_{T\to0} \int_0^\mu \frac{\tanh(t/2T)}{2t}\left(\frac{1}{\sqrt{\mu-t}}-\frac{1}{\sqrt{\mu}}\right) \dd t = \int_0^\mu \frac{1}{2t}\left(\frac{1}{\sqrt{\mu-t}}-\frac{1}{\sqrt{\mu}}\right) \dd t
= \frac{\ln 4}{2\sqrt{\mu}}
\end{equation}
as well as 
\begin{equation}
\lim_{T\to0} \int_0^\mu \frac{\tanh(t/2T)}{2t}\left(\frac{1}{\sqrt{\mu+t}}-\frac{1}{\sqrt{\mu}}\right) \dd t = \int_0^\mu \frac{1}{2t}\left(\frac{1}{\sqrt{\mu+t}}-\frac{1}{\sqrt{\mu}}\right) \dd t
= \frac{\ln \left(2(\sqrt{2}-1)\right)}{\sqrt{\mu}}.
\end{equation}
Using monotone convergence once more, we obtain
\begin{equation}
\lim_{T\to0} \int_\mu^\infty \frac{\tanh(t/2T)}{2t\sqrt{\mu+t}} \dd t = \int_\mu^\infty \frac{1}{2t\sqrt{\mu+t}} \dd t =\frac{\ln (\sqrt{2}+1)}{\sqrt{\mu}}.
\end{equation}
Combining all the terms we arrive at the first equality in \eqref{int_F_asy}.
Observe that
\begin{equation}
0<\int_\BR \chi_{\vert p \vert>\sqrt{2\mu}} F_{T,\mu}(p) \dd p\leq 2\int_{\sqrt{2\mu}}^\infty \frac{1}{p^2-\mu}\dd p <\infty.
\end{equation}
Therefore, this term is of order one for $T\to 0$ and $\int_\BR F_{T,\mu}(p)\dd p= \int_{-\sqrt{2\mu}}^{\sqrt{2\mu}} F_{T,\mu}(p)\dd p+O(1)$.
\end{proof}

\subsection{From Section~\ref{sec:weak}}\label{sec:pfweak}
\begin{proof}[Proof of Lemma~\ref{lea:estB1}]
In the case $x y >0$, the inequality follows immediately from the fact that $\vert \tanh(z)\vert<1$ for all $z\in \BR$.
In the case $xy<0$, let us replace $y\to -y$ and assume without loss of generality  that $x>y>0$.
Since the function $s \mapsto e^{-2s}$ is convex, we have 
\begin{equation}\label{exp_conv_est}
\frac{e^{-2y}-e^{-2x}}{x-y}\leq {}-\left.\frac{\dd }{\dd s}e^{-2s}\right \vert_{s=y}=2  e^{-2y}
\end{equation}
We estimate 
\begin{multline}
\frac{x+y}{x-y} (\tanh(x)-\tanh(y))
= \frac{2(x+y)}{1+e^{-2y}}\frac{e^{-2y}-e^{-2x}}{(x-y)(1+e^{-2x})} 
\leq \frac{2(x+y)e^{-2y}}{1+e^{-2y}}\min\left\{2, \frac{1}{x-y}\right\}\\
\leq \frac{4(2y+1/2)e^{-2y}}{1+e^{-2y}},
\end{multline}
where we maximized over $x$ in the last step.
The maximum of the last expression over $y$ is attained at the value $y=\tilde y$ satisfying $e^{-2\tilde y}=2\tilde y-1/2$.
Therefore, we get 
\begin{equation}
\frac{x+y}{x-y} (\tanh(x)-\tanh(y)) \leq 4(2\tilde y-1/2).
\end{equation}
The function $e^{-2 y}$ is decreasing in $y$ and $2y-1/2$ is increasing.
For $y=1/2$ we have $e^{-1}<1/2$, hence the intersection point $\tilde y$ satisfies $0 < \tilde y < 1/2$ .
Thus, $\frac{x+y}{x-y} (\tanh(x)-\tanh(y))<2$, which proves the claim.
\end{proof}

\begin{proof}[Proof of Lemma~\ref{lea:estB2}]
Without loss of generality, we may assume that $y<x$.
We have 
\begin{multline*}
\tanh(x)-\tanh(y)=\frac{e^x-e^{-x}}{e^x+e^{-x}}-\frac{e^y-e^{-y}}{e^y+e^{-y}}=2\frac{e^{x-y}-e^{y-x}}{(e^x+e^{-x})(e^y+e^{-y})}\\
\leq 2\frac{e^{x-y}-e^{y-x}}{e^{x+y}}=2(e^{-2y}-e^{-2x})
\end{multline*}
Applying \eqref{exp_conv_est} the claim follows.
\end{proof}

\subsection{From Section~\ref{sec:strong}}\label{sec:pfstrong}
\begin{proof}[Proof of Lemma~\ref{highTlimineq2}]
By concavity of $\tanh(x)$ for $x\geq0$ for $x,y\geq 0$ it holds that
\begin{equation}
\frac{\tanh(x)+\tanh(y)}{2}\leq \tanh \left(\frac{x+y}{2}\right) \Leftrightarrow \frac{\tanh(x)+\tanh(y)}{2(x+y)}\leq \frac{\tanh \left(\frac{x+y}{2}\right)}{x+y}
\end{equation}
Choosing $x=(p+q)^2/8$ and $y=(p-q)^2/8$ gives the desired inequality.
\end{proof}


\section*{Acknowledgments} 
We thank Egor Babaev for encouraging us to study this problem, and Rupert Frank for many fruitful discussions. 
Funding from the European Union's Horizon 2020 research and innovation programme under the ERC grant agreement No. 694227 (B.R. and R.S.) is gratefully acknowledged.

\bibliographystyle{abbrv}

\begin{thebibliography}{10}

\bibitem{abrikosov_concerning_1964}
A.~A. Abrikosov.
\newblock Concerning {Surface} {Superconductivity} in {Strong} {Magnetic}
  {Fields}.
\newblock {\em J. Exptl. Theoret. Phys. (U.S.S.R)}, 47(2):720--733, 1964.

\bibitem{adams_sobolev_2003}
R.~A. Adams and J.~J. Fournier.
\newblock {\em Sobolev {Spaces}}, volume 140 of {\em Pure and {Applied}
  {Mathematics}}.
\newblock Academic Press, 2nd edition, 2003.

\bibitem{barkman_elevated_2022}
M.~Barkman, A.~Samoilenka, A.~Benfenati, and E.~Babaev.
\newblock Elevated critical temperature at {BCS} superconductor-band insulator interfaces.
\newblock {\em arXiv:2201.11614 [cond-mat]}, Jan. 2022.

\bibitem{benfenati_boundary_2021}
A.~Benfenati, A.~Samoilenka, and E.~Babaev.
\newblock Boundary effects in two-band superconductors.
\newblock {\em Physical Review B}, 103(14):144512, Apr. 2021.

\bibitem{caroli_sur_1962}
C.~Caroli, P.~De~Gennes, and J.~Matricon.
\newblock Sur certaines propri\'{e}t\'{e}s des alliages supraconducteurs non
  magn\'{e}tiques.
\newblock {\em Journal de Physique et le Radium}, 23(10):707--716, 1962.

\bibitem{de_gennes_boundary_1964}
P.~G. De~Gennes.
\newblock Boundary {Effects} in {Superconductors}.
\newblock {\em Reviews of Modern Physics}, 36(1):225--237, Jan. 1964.

\bibitem{egger_bound_2020}
S.~Egger, J.~Kerner, and K.~Pankrashkin.
\newblock Bound states of a pair of particles on the half-line with a general
  interaction potential.
\newblock {\em Journal of Spectral Theory}, 10(4):1413--1444, Dec. 2020.

\bibitem{frank_bcs_2019}
R.~L. Frank, C.~Hainzl, and E.~Langmann.
\newblock The {BCS} critical temperature in a weak homogeneous magnetic field.
\newblock {\em Journal of Spectral Theory}, 9(3):1005--1062, Mar. 2019.

\bibitem{frank_microscopic_2011}
R.~L. Frank, C.~Hainzl, R.~Seiringer, and J.~P. Solovej.
\newblock Microscopic {Derivation} of {Ginzburg}-{Landau} {Theory}.
\newblock {\em J. Amer. Math. Soc.} 25, 667, 2012.

\bibitem{frank_condensation_2017}
R.~L. Frank, M.~Lemm, and B.~Simon.
\newblock Condensation of fermion pairs in a domain.
\newblock {\em Calculus of Variations and Partial Differential Equations},
  56(2):54, Apr. 2017.

\bibitem{gennes_superconductivity_1999}
P.-G.~d. Gennes.
\newblock {\em Superconductivity of metals and alloys}.
\newblock Advanced book classics. Advanced Book Program, Perseus Books,
  Reading, Mass, 1999.

\bibitem{hainzl_bcs_2008}
C.~Hainzl, E.~Hamza, R.~Seiringer, and J.~P. Solovej.
\newblock The {BCS} {Functional} for {General} {Pair} {Interactions}.
\newblock {\em Communications in Mathematical Physics}, 281(2):349--367, July
  2008.

\bibitem{hainzl_critical_2008}
C.~Hainzl and R.~Seiringer.
\newblock Critical temperature and energy gap for the {BCS} equation.
\newblock {\em Physical Review B}, 77(18):184517, May 2008.

\bibitem{parks_superconductivity_1969}
R.~D. Parks, editor.
\newblock {\em Superconductivity. 1}.
\newblock Dekker, New York, 1969.

\bibitem{roos_two-particle_2021}
B.~Roos and R.~Seiringer.
\newblock Two-{Particle} {Bound} {States} at {Interfaces} and {Corners}.
\newblock {\em Journal of Functional Analysis}, 202(12):109455, June 2022.

\bibitem{samoilenka_boundary_2020}
A.~Samoilenka and E.~Babaev.
\newblock Boundary states with elevated critical temperatures in
  {Bardeen}-{Cooper}-{Schrieffer} superconductors.
\newblock {\em Physical Review B}, 101(13):134512, Apr. 2020.

\bibitem{samoilenka_microscopic_2021}
A.~Samoilenka and E.~Babaev.
\newblock Microscopic derivation of superconductor-insulator boundary
  conditions for {Ginzburg}-{Landau} theory revisited: {Enhanced}
  superconductivity at boundaries with and without magnetic field.
\newblock {\em Physical Review B}, 103(22):224516, June 2021.

\end{thebibliography}

\end{document}